\documentclass[runningheads]{llncs}
\usepackage{graphicx}

\usepackage[utf8]{inputenc}
\usepackage[english]{babel}
\usepackage[margin=1in]{geometry}
\usepackage[fontsize=11pt]{scrextend}

\usepackage{nicefrac}

\usepackage{float}
\usepackage{graphicx}
\usepackage{subcaption}

\usepackage{amssymb}
\usepackage{amsmath}
\usepackage[boxed]{algorithm2e}
\usepackage{bm,bbm}

\usepackage[dvipsnames,usenames]{xcolor}
\usepackage[colorlinks=true,pdfpagemode=UseNone,urlcolor=RoyalBlue,linkcolor=RoyalBlue,citecolor=OliveGreen,pdfstartview=FitH]{hyperref}

\usepackage{tcolorbox}
\usepackage{todonotes}

\usepackage{booktabs}

\usepackage{footnote}
\makesavenoteenv{tabular}
\makesavenoteenv{table}

\renewcommand{\qed}{\hfill$\square$}

\newcommand{\budget}{\theta}
\newcommand{\load}{\ell}
\newcommand{\rank}{\rho}
\newcommand{\threshold}{\tau}
\newcommand{\before}{\prec}
\newcommand{\critical}{\mu}

\newcommand{\OPT}{\text{\rm OPT}}
\newcommand{\SOPT}{\text{\rm S-OPT}}

\DeclareMathOperator*{\maximize}{\textrm{maximize}}
\DeclareMathOperator*{\minimize}{\textrm{minimize}}

\newcommand{\E}{\mathbf{E}}

\newcommand*\D{\,\mathrm{d}}

\begin{document}
\title{
    Online Matching with Stochastic Rewards:\\
    Advanced Analyses Using Configuration Linear Programs
}
\titlerunning{Online Matching with Stochastic Rewards: Advanced Analyses}
%


\author{
   Zhiyi Huang \inst{1}
   \and
   Hanrui Jiang \inst{2}
   \and
   Aocheng Shen \inst{2}
   \and
   Junkai Song  \inst{1}
   \and
   Zhiang Wu  \inst{3}
   \and
   Qiankun Zhang  \inst{2} 
}

\institute{
    The University of Hong Kong  \\ \email{zhiyi@cs.hku.hk}, \email{dsgsjk@connect.hku.hk}\and
    Huazhong University of Science and Technology  \\ \email{hanry@hust.edu.cn}, \email{awsonshen@hust.edu.cn},
    \email{qiankun@hust.edu.cn}\\
    \and
    Hong Kong University of Science and Technology \\
    \email{zwube@connect.ust.hk}
}

\authorrunning{Z. Huang et al.}

\maketitle              
\begin{abstract}
    \thispagestyle{empty}
    Mehta and Panigrahi (2012) proposed Online Matching with Stochastic Rewards, which generalizes the Online Bipartite Matching problem of Karp, Vazirani, and Vazirani (1990) by associating the edges with success probabilities.
This new feature captures the pay-per-click model in online advertising.
Recently, Huang and Zhang (2020) studied this problem under the online primal dual framework using the Configuration Linear Program (LP), and got the best known competitive ratios of the Stochastic Balance algorithm.
Their work suggests that the more expressive Configuration LP is more suitable for this problem than the Matching LP.

This paper advances the theory of Configuration LP in two directions.
Our technical contribution includes a characterization of the joint matching outcome of an offline vertex and \emph{all its neighbors}.
This characterization may be of independent interest, and is aligned with the spirit of Configuration LP.
By contrast, previous analyses of Ranking generally focus on only one neighbor.
Second, we designed a Stochastic Configuration LP that captures a stochastic benchmark proposed by Goyal and Udwani (2020), who used a Path-based LP.
The Stochastic Configuration LP is smaller and simpler than the Path-based LP.
Moreover, using the new LP we improved the competitive ratio of Stochastic Balance from $0.596$ to $0.611$ when the success probabilities are infinitesimal, and to $0.613$ when the success probabilities are further equal.

\end{abstract}

\section{Introduction}
\label{sec:intro}

Suppose that Alice is planning an upcoming trip to Hawaii and she just types ``Hawaii resort'' in a search engine.
Once she presses the return key, the search engine will provide a list of related websites, among which several websites have a small ``Ad'' tag next to them to indicate that they are sponsored results.
We will refer to such websites as advertisers.
If Alice clicks on a sponsored result, the advertiser would pay the search engine, where the amount depends on its bid for the keywords among other factors.
This is one of many scenarios of online advertising, which contributes hundreds of billions of US dollars to the annual revenue of major IT companies.

Online advertising presents many unique challenges, each of which has led to a long line of research.
For example, the search engine must decide which advertiser shall be returned to a search query without knowing what search queries may come next.
The essence of this online decision-making problem is captured by Online Bipartite Matching, which amazingly was introduced by Karp, Vazirani, and Vazirani~\cite{KarpVV:STOC:1990} well before online advertising even existed.
To see the connection, consider the advertisers and search queries as vertices on the two sides of a bipartite graph.
The advertisers are known from the beginning and we call them \emph{offline vertices}.
The search queries come one by one and we call them \emph{online vertices}.
The algorithm needs to decide how to match each online vertex on its arrival, with the goal of maximizing the matching size.
We measure an algorithm by the worst-case ratio of its expected matching size to the size of the optimal matching in hindsight, a.k.a., the \emph{competitive ratio}.

Readers may notice a disparity, however, if they compare our example scenario of online advertising with the Online Bipartite Matching problem.
In online advertising, the search engine cannot control whether a user clicks on the ad or not.
The best that a search engine could do is to estimate how likely the user will click, based on the keywords and the information it gathers about the user.
To this end, an attempt to match an online vertex to an offline neighbor only succeeds with some probability, rather than with certainty like in Online Bipartite Matching.
Because of this disparity, Mehta and Panigrahi~\cite{MehtaP:FOCS:2012} proposed Online Matching with Stochastic Rewards, extending Online Bipartite Matching by associating the edges with success probabilities.

Much progress has been made in Online Matching with Stochastic Rewards in the last decade.
Most related to this paper is the work of Huang and Zhang~\cite{HuangZ:STOC:2020}, who successfully applied the online primal-dual framework to this problem.
Their analysis of the Stochastic Balance algorithm, which matches each online vertex to a neighbor that has been matched the least number of times so far, yielded the best competitive ratios to date. 
Conceptually, they found that the Matching Linear Program (LP) is not expressive enough for an online primal-dual analysis of this problem, and instead one could consider the Configuration LP.

There are still unanswered questions, however, that cast some doubt on the usefulness of the Configuration LP.
For instance, Mehta and Panigrahi~\cite{MehtaP:FOCS:2012} showed that the Ranking algorithm, which randomly ranks the offline vertices at the beginning and matches each online vertex to the neighbor with the highest rank, is another competitive algorithm for Online Matching with Stochastic Rewards.
However, Huang and Zhang~\cite{HuangZ:STOC:2020} failed to offer any non-trivial competitive analysis for the Ranking algorithm using the online primal-dual analysis with Configuration LP.
Moreover, Goyal and Udwani~\cite{GoyalU:OR:2022} introduced an alternative benchmark, which we will refer to as the stochastic benchmark (see Section~\ref{sec:prelim} for details).
They proved that the competitive ratios of both Stochastic Balance and Ranking would be strictly larger if they were compared against the stochastic benchmark.
To do so, they designed a Path-based LP which is completely different than the Configuration LP.
Can we reproduce or even improve these results by further developing the theory of online primal-dual with Configuration LP?

\subsection{Our Results}

%
Our first contribution is \emph{an online primal-dual analysis of Ranking based on the Configuration LP}, improving the competitive ratio of Ranking from $0.534$ to $0.572$.
To present the technical novelty that enables this analysis, we need to first explain the intuition behind the usefulness of the Configuration LP in an online primal-dual analysis compared to the simple Matching LP.
The essence of an online primal-dual analysis is to design an appropriate gain-splitting rule.
When the algorithm matches an online vertex $v$ to an offline neighbor $u$ and increases the expected objective accordingly by success probability $p_{uv}$, imagine that we further split the gain of $p_{uv}$ between its two vertices $u$ and $v$.
We need the gain splitting rule to satisfy some LP-dependent conditions which essentially say that the total expected gain of neighboring vertices are sufficiently large.
More precisely, if we use the Matching LP, the condition is: 
\begin{equation}
    \label{eqn:intro-matching-lp-condition}
    p_{uv} \cdot \E \big[ \text{gain of $u$} \big] + \E \big[ \text{gain of $v$} \big] \ge \Gamma \cdot p_{uv}
    ~,
\end{equation}
where $\Gamma$ is the competitive ratio that we seek to prove for the algorithm.

By contrast, if we use the Configuration LP, the critical condition no longer considers just a single edge $(u, v)$.
Instead, it examines an offline vertex $u$ and \emph{a subset of neighbors} $S$ whose success probabilities sum to about $1$, and requires that:
\begin{equation}
    \label{eqn:intro-config-lp-condition}
    \E \big[ \text{gain of $u$} \big] + \sum_{v \in S} \E \big[ \text{gain of $v$} \big] \ge \Gamma
    ~.
\end{equation}

In other words, the configuration LP only needs an amortized version of Eqn.~\eqref{eqn:intro-matching-lp-condition}, which sums over $v \in S$.
Even if Eqn.~\eqref{eqn:intro-matching-lp-condition} fails in the worst-case scenario with respect to (w.r.t.)\ an offline vertex $u$ and a single neighbor $v$, Eqn.~\eqref{eqn:intro-config-lp-condition} may still hold if such worst-case cannot happen to all $v \in S$ at the same time. 
Note that, however, to use the power of the weaker condition, we need to characterize the joint matching outcome of $u$ and a subset of neighbors $S$, rather than doing so for each $v \in S$ separately.
Huang and Zhang~\cite{HuangZ:STOC:2020} gave such a characterization for Stochastic Balance.

This paper finds that even the joint matching outcome of $u$ and $S$ is insufficient for getting a useful characterization of the worst-case scenario of Eqn.~\eqref{eqn:intro-config-lp-condition} w.r.t.\ Ranking.
Instead, we consider the joint outcome of $u$ and \emph{all its neighbors}, including those outside $S$.
Characterizing the entire neighborhood's matching outcome may seem counter-intuitive \emph{a priori}, which may be why Huang and Zhang~\cite{HuangZ:STOC:2020} failed to get a non-trivial competitive ratio for Ranking.
The characterization and the new approach in general may be useful for other online matching problems whose offline vertices can be matched more than once, e.g., the AdWords problem~\cite{MehtaSVV:JACM:2007}.

%
Our second contribution is \emph{a new Stochastic Configuration LP}, which we design to bound the stochastic benchmark of Goyal and Udwani~\cite{GoyalU:OR:2022}.
The stochastic benchmark refers to the best objective achievable by an algorithm which knows the graph and the edges' success probabilities but not the realization of edge successes and failures, and which still needs to match online vertices by their arrival order.
It is natural to consider LPs whose decision variables depend on the observed edge successes and failures to capture the stochastic benchmark.
This is the approach of Goyal and Udwani~\cite{GoyalU:OR:2022}, who referred to the observed edge successes and failures as a sample path and gave a Path-based LP.
Their LP is huge, however, because the number of sample paths is exponential in the number of edges by definition.
Further, compared to the Matching and Configuration LPs, the Path-based LP has an extra set of consistency constraints: if a sample path is a sub-path of another sample path, then the relevant matching outcome must be consistent.
These constraints lead to new dual variables that are outside the scope of existing gain-splitting methods.
Goyal and Udwani~\cite{GoyalU:OR:2022} still found a canonical gain-splitting rule and sufficient conditions for proving competitive ratios for Ranking and Stochastic Balance, but their conditions do not correspond to the constraints of Path-based LP's dual.
Therefore, it is unclear how to apply their approach to other problems.

We use a different strategy to design Stochastic Configuration LP for the stochastic benchmark.
Instead of using variables $x_{uS}$ for the probability an offline vertex $u$ is matched to a subset of online vertices $S$ as in the Configuration LP, we consider variables $y_{uS}$ that represents the probability an offline vertex $u$ \emph{would be matched to subset $S$ if all matches to $u$ failed}.
The size of the Stochastic Configuration LP is therefore exponential only in the number of vertices, like the original Configuration LP.
Further, these variables implicitly capture the consistency requirements and therefore no additional consistency constraints are needed.
As a result, the online primal dual analysis can directly work with the (approximate) dual constraints.
Compared to condition \eqref{eqn:intro-config-lp-condition} given by the Configuration LP, the condition w.r.t.\ the stochastic variant scales the gains of online vertices and the competitive ratio on the right-hand-side by some multipliers smaller than $1$, with the multiplier of competitive ratio being smaller.
Hence, the resulting condition is indeed easier to satisfy.

Using the Stochastic Configuration LP, we improved the competitive ratios of Stochastic Balance w.r.t.\ the stochastic benchmark from $0.596$ to $0.611$ when the success probabilities are infinitesimal, and to $0.613$ if the success probabilities are further equal.%
\footnote{
We remark that the competitive ratio of Goyal and Udwani~\cite{GoyalU:OR:2022} holds for the optimal value of the Path-based LP, which may be strictly larger than the optimal value of our Stochastic Configuration LP in some cases.
Nevertheless, both LPs' optimal values upper bound the stochastic benchmark.
}%
\footnote{We thank an anonymous reviewer for pointing out that a preliminary arXiv version (v5) of \cite{GoyalU:OR:2022} claims a competitive ratio of $0.605$ for the general infinitesimal case without proof. Nevertheless, our bound is also larger than that.}
We also reproduce the optimal $1-\frac{1}{e}$ competitive ratio of Ranking w.r.t.\ the stochastic benchmark.

\begin{table}[t]
    \centering
    \caption{Summary of Results.
        The non-stochastic and stochastic benchmarks are denoted as $\OPT$ and $\SOPT$ respectively.
        The results for Stochastic Balance apply to infinitesimal success probabilities.
        We round the competitive ratios down to the third digit after the decimal point.
    }
    \label{tab:summary}

    \begin{tabular}{p{2.1in}llllll}
        \toprule
        & \multicolumn{3}{p{1.5in}}{$\OPT$} & \multicolumn{3}{p{1.5in}}{$\SOPT$} \\
        \midrule
        Ranking {\scriptsize (Equal Prob.)} & $0.534$ \cite{MehtaP:FOCS:2012} & $\to$ & $\bm{0.572}$
        & $1-\frac{1}{e} \approx 0.632$ \cite{GoyalU:OR:2022} & \\
        Stochastic Balance {\scriptsize (Equal Prob.)} & $0.576$ \cite{HuangZ:STOC:2020} & & & $0.596$ \cite{GoyalU:OR:2022} & $\to$ & $\bm{0.613}$ \\
        Stochastic Balance {\scriptsize (General)} & $0.572$ \cite{HuangZ:STOC:2020}   & & & $0.596$    \cite{GoyalU:OR:2022}  & $\to$ & $\bm{0.611}$ \\
        \bottomrule
    \end{tabular}
\end{table}

\subsection{Related Works}

Our analyses follow the (randomized) online primal dual framework by Devanur, Jain, and Kleinberg~\cite{DevanurJK:SODA:2013}, who were inspired by Birnbaum and Mathieu~\cite{BirnbaumM:2008} and Goel and Mehta~\cite{GoelM:SODA:2008}.

An online advertising platform's objective is rarely as simple as maximizing the matching size.
For instance, advertisers usually have different values for different keywords.
Hence, the literature has subsequently introduced many variants of the original Online Bipartite Matching problem by Karp et al.~\cite{KarpVV:STOC:1990}.
Aggarwal et al.~\cite{AggarwalGKM:SODA:2011} generalized Ranking and the optimal $1-\frac{1}{e}$ competitive ratio to the (offline) vertex-weighted problem.
Feldman et al.~\cite{FeldmanKMMP:WINE:2009} considered the edge-weighted problem with free disposal, and gave a $1-\frac{1}{e}$ competitive algorithm assuming large capacities of offline vertices;
see also Devanur et al.~\cite{DevanurHKMY:TEAC:2016} for an online primal dual version of this result.
Their algorithm may be viewed as a generalization of the Balance algorithm for unweighted matching by Kalyanasundaram and Pruh~\cite{KalyanasundaramP:TCS:2000}, which also inspired the Stochastic Balance algorithm by Mehta and Panigrahi~\cite{MehtaP:FOCS:2012} for Online Matching with Stochastic Rewards.
Recently, Fahrbach et al~.\cite{FahrbachHTZ:JACM:forthcoming} gave the first algorithm that is strictly better than $\frac{1}{2}$-competitive for the edge-weighted problem without large-capacity assumption.
The ratio was later improved independently to $0.509$ by Shin and An~\cite{ShinA:ISAAC:2021}, to $0.519$ by Gao et al.~\cite{GaoHHNYZ:FOCS:2021}, and to $0.536$ by Blanc and Charikar~\cite{BlancC:FOCS:2021}.
Another important variant is the AdWords problem by
Mehta et al.~\cite{MehtaSVV:JACM:2007}, where each offline vertex may be matched multiple times and has a budget-additive valuation.
Under a small bid assumption, Mehta et al.~\cite{MehtaSVV:JACM:2007} gave an optimal $1-\frac{1}{e}$ competitive algorithm.
See also Buchbinder, Jain, and Naor~\cite{BuchbinderJN:ESA:2007} for an online primal dual analysis of this algorithm, and Devanur and Jain~\cite{DevanurJ:STOC:2012} for a generalization called Online Matching with Concave Returns.
Huang, Zhang, and Zhang~\cite{HuangZZ:FOCS:2020} recently broke the $\frac{1}{2}$ barrier for AdWords with general bids and gave a $0.501$-competitive algorithm, using an approach similar to Fahrbach et al.~\cite{FahrbachHTZ:JACM:forthcoming}.
A recent line of results~\cite{Liang2022OnTP,Udwani2021AdwordsWU,Vazirani2022TowardsAP} studied budget-oblivious algorithms for AdWords, whose allocation policies do not depend on the budgets of offline vertices.

Besides the stochastic edge successes and failures in Online Matching with Stochastic Rewards, the literature has also studied other stochastic models of online matching problems.
Mahdian and Yan~\cite{MahdianY:STOC:2011} and Karande, Mehta, and Tripathi~\cite{KarandeMT:STOC:2011} independently showed that Ranking is strictly better than $1-\frac{1}{e}$ competitive if online vertices arrive by random order.
Huang et al.~\cite{HuangTWZ:TALG:2019} extended this result to the vertex-weighted problem, which was further improved upon by Jin and Williamson~\cite{JinW:WINE:2021}.
The problem with stochastically generated graphs was first studied by Feldman et al.~\cite{FeldmanMMM:FOCS:2009}, who named it Online Stochastic Matching.
We also refer readers to the most recent works on this problem by Huang and Shu~\cite{HuangS:STOC:2021}, Huang, Shu, and Yan~\cite{HuangSY:STOC:2022}, and Tang, Wu, and Wu~\cite{TangWW:STOC:2022}.

\subsection{Paper Outline}

Section~\ref{sec:prelim} presents a formal definition of online matching with stochastic rewards, the benchmarks we will compare against, and existing linear programs and algorithms. Section~\ref{sec:lp} introduces our new stochastic Configuration LPs, and their properties.
Section~\ref{sec:ranking} shows the online primal-dual analyses of Ranking.
Finally, our results for Stochastic Balance are presented in Appendix~\ref{sec:balance} because of page limitation and the fact that the improvements mainly come from a combination of the structural lemmas by Huang and Zhang~\cite{HuangZ:STOC:2020} and the new stochastic Configuration LP proposed in this paper.

\section{Preliminaries}
\label{sec:prelim}

%
We write $x^+$ for function $\max\,\{ x, 0 \}$.
For any vector $x$ and any index $i$, we write $x_{-i}$ for the vector with the $i$-th entry removed, and $(x_i, x_{-i})$ for a vector whose $i$-th entry is $x_i$, and whose other entries are $x_{-i}$.

\subsection{Model}
Consider a bipartite graph $G = (U, V, E)$.
Vertices in $U$ are referred to as \emph{offline vertices} because they are known to the algorithm from the beginning.
Vertices in $V$ arrive one by one, and are referred to as \emph{online vertices}.
We will write $v \before v'$ to denote that $v$ arrives before $v'$.
Each edge $(u, v)$ is associated with a success probability $0 \le p_{uv} \le 1$;
further define $p_{uS} = \sum_{v \in S} p_{uv}$ for any $u \in U$ and any $S \subseteq V$.
When an online vertex $v \in V$ arrives, the algorithm sees $v$'s incident edges and the corresponding success probabilities.
Then, the algorithm makes an irrevocable matching decision about $v$.
Should the algorithm choose to match $v$ to an offline vertex $u$, the match would succeed with probability $p_{uv}$.
If a match is unsuccessful, the offline vertex $u$ can still be matched to future online vertices, but the online vertex $v$ will not get a second chance.
An offline vertex $u \in U$ is \emph{successful} if there is an edge that is matched to $u$ and is successful;
$u$ is \emph{unsuccessful} otherwise.
We want to maximize the number of successful offline vertices.

We remark that the results in this paper generalize to the vertex-weighted problem, where offline vertices have positive weights and we want to maximize the total weight of successful offline vertices.
Readers familiar with the online matching literature shall not find this surprising since almost all known results for unweighted online matching problems generalize to the vertex-weighted problems.
This version of the paper only presents the unweighted case to keep the exposition simple.

By allowing $p_{uv} = 0$, we may assume without loss of generality that $E = U \times V$ which will simplify the exposition in some parts of this paper.
We say that $u$ and $v$ are neighbors if $p_{uv} > 0$.
We say that the success probabilities are equal if $p_{uv}$ is either $0$ and $p$ for some $0 \le p \le 1$.

\smallskip
\noindent
\emph{Stochastic Budgets.~}
Mehta and Panigrahi~\cite{MehtaP:FOCS:2012} showed that when the success probabilities are infinitesimal, we can consider an equivalent model in which the randomness plays a different role.
At the beginning, each offline vertex $u \in U$ independently draws a budget $\budget_u$ from the exponential distribution with mean $1$.
When an online vertex $v \in V$ arrives, the algorithm may match it to any offline neighbor and collect a gain that equals the success probability of the edge.
However, the total gain from an offline vertex $u $ is capped by its budget $\theta_u$, i.e., it is either the sum of success probabilities of the edges matched to $u$, denoted as $\load_u$ and referred to as its \emph{load}, or its budget $\budget_u$, whichever is smaller.
Further, the budget $\budget_u$ is unknown to the algorithm until the moment when the load $\load_u$ exceeds the budget, which corresponds to $u$'s succeeding.

\smallskip
\noindent
\emph{Benchmarks and Competitive Analysis.~}
The \emph{offline (non-stochastic) optimum}, denoted as $\OPT$, refers to the optimal objective achievable by a computationally unlimited offline algorithm that knows the graph and success probabilities, and when the budgets are non-stochastically equal to $1$.
The offline optimum upper bounds the objective achievable by any (online or offline) algorithm in the model with stochastic budgets~\cite{MehtaP:FOCS:2012}.

The \emph{offline stochastic optimum}, denoted as $\SOPT$, refers to the optimal objective achievable by a computationally unlimited offline algorithm that knows the graph and success probabilities, but can only match the online vertices one by one by their arrival order, and can only observe the stochastic budget of an offline vertex when its load exceeds the budget like online algorithms.

An online algorithm is $\Gamma$-competitive w.r.t.\ one of these benchmarks if for \emph{any instance} of Online Matching with Stochastic Rewards, the expected objective given by the algorithm is at least a $\Gamma$ fraction of the benchmark.
Here $0 \le \Gamma \le 1$ is called the \emph{competitive ratio}.

\subsection{Existing Linear Programs}

%
Consider $0 \le x_{uv} \le 1$ as the probability that $u$ is matched to $v$ (successful or not) for any $u \in U$ and any $v \in V$.
The \emph{Matching LP} is defined as:
\begin{align*}
    \max_{x \ge 0} \quad & \sum_{u \in U} \sum_{v \in V} p_{uv} x_{uv}                                       \\
    \textrm{s.t.} \quad & \sum_{v \in V} p_{uv} x_{uv} \le 1          &  & \forall u \in U                  \\
    & \sum_{u \in U} x_{uv} \le 1                 &  & \forall v \in V                  
\end{align*}

The first constraint states that the expected load of an offline vertex $u$, which is equal to the probability that $u$ succeeds, is at most $1$.
The second constraint says that each online vertex $v$ can be matched to at most one offline vertex.
An offline allocation yields a feasible solution to the Matching LP by the aforementioned interpretation of $x_{uv}$.
Hence, the optimal objective of Matching LP upper bounds the offline optimum $\OPT$.

%
For any offline vertex $u \in U$ and any subset of online vertices $S \subseteq U$, let $\bar{p}_{uS} = \min \{ p_{uS}, 1 \}$, where recall that $p_{uS} = \sum_{v \in S} p_{uv}$.
Consider $0 \le x_{uS} \le 1$ as the probability that $S$ is the subset of online vertices matched to $u$ by the algorithm.
Huang and Zhang~\cite{HuangZ:STOC:2020} used the following \emph{Configuration LP} and its dual LP:
\begin{align}
    \max_{x \ge 0} \quad & \sum_{u \in U} \sum_{S \subset V} \bar{p}_{uS} x_{uS} && \textbf{(Primal)}
    \notag \\
    \textrm{s.t.} \quad & \sum_{S \subseteq V} x_{uS} \le 1                         &                                                                       & \forall u \in U \notag \\
    & \sum_{u \in U} \sum_{S \subseteq V: v \in S} x_{uS} \le 1 &                                                                       & \forall v \in V \notag \\[1ex]
%
%
    %
    \min_{\alpha, \beta \ge 0} \quad & \sum_{u \in U} \alpha_u + \sum_{v \in V} \beta_v && \textbf{(Dual)} \notag \\
    \textrm{s.t.} \quad &
    \alpha_u + \sum_{v \in S} \beta_v \ge \bar{p}_{uS} &                                                           & \forall u \in U, \forall S \subseteq V \label{eqn:configuration-dual}
\end{align}

The optimal objective of Configuration LP also upper bounds the offline optimum $\OPT$.
In fact, it coincides with the optimal objective of Matching LP when the success probabilities are infinitesimal (c.f., Huang and Zhang~\cite{HuangZ:STOC:2020}).
Nevertheless, its dual structure is better suited for an online primal dual analysis.

\begin{lemma}[c.f., Lemma 3 of Huang and Zhang~\cite{HuangZ:STOC:2020}]
    \label{lem:configuration-lp-primal-dual}
    Suppose that an algorithm for Online Matching with Stochastic Rewards can further split $p_{uv}$ between $\alpha_u$ and $\beta_v$ each time it matches an edge $(u,v)$, such that the dual constraint~\eqref{eqn:configuration-dual} holds up to factor $\Gamma$ \emph{in expectation}:
    \[
        \E \Big[ \, \alpha_u + \sum_{v \in S} \beta_v \, \Big] ~\ge~ \Gamma \cdot \bar{p}_{uS}
        ~.
    \]
    Then, the algorithm is $\Gamma$-competitive w.r.t.\ $\OPT$.
\end{lemma}

%
Finally, Goyal and Udwani~\cite{GoyalU:OR:2022} introduced the \emph{Path-based LP} which is tailored for the offline stochastic optimum $\SOPT$.
A \emph{sample path} refers to all information of what has happened so far, including the subset of edges matched by the algorithm, and the realization of whether these matches are successful.
This LP considers variable $0 \le x_{uv}^\omega \le 1$ for any edge $(u,v)$ which represents the probability that the algorithm matches $v$ to $u$ conditioned on a sample path $\omega$.
We omit the details of Path-based LP because it is not directly related to the LPs in this paper, and its constraints are too complicated to be covered concisely here.
We remark that the number of possible sample paths and thus the number of variables of Path-based LP are exponential in the number of \emph{edges}.
By contrast, the number of variables of Configuration LP is exponential in the number of \emph{vertices}, which is typically much smaller than the number of edges.

\subsection{Existing Algorithms}

\noindent
\textbf{Ranking.~}
We will consider the \emph{Ranking} algorithm in the case of equal success probabilities.
At the beginning, draw a rank $\rank_u \in [0, 1]$ uniformly at random for each offline vertex $u \in U$.
Then, on the arrival of each online vertex $v \in V$, match it to the unsuccessful neighbor that has the smallest rank.
We can break ties arbitrarily since they occur with zero probability.

\smallskip
\noindent
\textbf{Stochastic Balance (Equal Probability).~}
On the arrival of each online vertex $v \in V$, match it to the unsuccessful neighbor that has the smallest load, breaking ties arbitrarily, e.g., by the lexicographical order.
Recall that the load of an offline vertex is the sum of success probabilities of the edges that have been matched to it so far.

\smallskip
\noindent
\textbf{Stochastic Balance (General).~}
We first describe a fractional algorithm.
On the arrival of each online vertex $v \in V$, consider a continuous process that matches $v$ fractionally to the unsuccessful neighbor $u \in U$ with the largest $p_{uv} \big(1 - g(\load_u) \big)$ where $g : [0, 1] \to [0, 1]$ is a non-decreasing function chosen by the algorithm designer.
If the success probabilities are infinitesimal, one can further convert it to an integral algorithm with no loss in the competitive ratio, through a reduction by Huang and Zhang~\cite{HuangZ:STOC:2020} based on randomized rounding.
Hence, it suffices to analyze the fractional algorithm's competitive ratio.

\section{Stochastic Configuration Linear Programs}
\label{sec:lp}

\subsection{Stochastic Thresholds}

Even when the success probabilities are arbitrary, we may still consider an equivalent model that generalizes the viewpoint of stochastic budgets for infinitesimal success probabilities.
We call this the \emph{stochastic thresholds viewpoint}.
At the beginning, each offline vertex $u$ independently draws a threshold $\threshold_u$ uniformly from $[0, 1]$.
When an online vertex $v$ arrives, the algorithm may match it to any offline neighbor and collect a gain which equals the success probability of the edge.
For each offline vertex $u$ and any subset of online vertices $S$, define:
\[
    \tilde{p}_{uS} = 1 - \prod_{v \in S} \big(1 - p_{uv}\big)
    ~.
\]

An offline vertex $u$ is successful (and can no longer be matched) if the subset of online vertices $S$ matched to it satisfies $\tilde{p}_{uS} \ge \threshold_u$.
Further, the threshold $\threshold_u$ is unknown to the algorithm until the moment that offline vertex $u$ becomes successful.

Observe that when the success probabilities are infinitesimal, we have $\tilde{p}_{uS} = 1 - e^{- p_{uS}}$, and the stochastic budget and threshold of an offline vertex $u \in U$ satisfy $\budget_u = - \ln(1- \tau_u)$.
%

\subsection{Reduced-form Stochastic Configuration Linear Program}

Consider $0 \le y_{uS} \le 1$ as the probability that a subset of online vertices $S$ would be matched to an offline vertex $u$ if it has an infinite stochastic budget $\theta_u = \infty$ (i.e., when $\threshold_u = 1$ if the success probabilities are not infinitesimal), over the randomness of the stochastic budgets $\budget_{-u}$ of other offline vertices and the randomness of the algorithm.
Further, for any $S \subseteq V$ and $v \in V$, let:
\[
    S(v) = \big\{ v' \in S : v' \before v \big\}
\]
denote the subset of online vertices in $S$ that arrive before $v$.
We consider the following \emph{Reduced-form Stochastic Configuration LP} and its dual:
\begin{align}
    \max_{y \ge 0} \quad                                                                                  &
    \sum_{u \in U} \sum_{S \subseteq V} \tilde{p}_{uS} \, y_{uS}                                     &                         & \textbf{(Primal)} \notag                                                                                             \\
    \textrm{s.t.} \quad                                                                                        &
    \sum_{S \subseteq V} y_{uS} \le 1                                                                &                         & \forall u \in U \label{eqn:reduced-form-lp-offline}                                                                  \\
                                                                                                     &
    \sum_{u \in U} \sum_{S \subseteq V : v \in S} \big( 1 - \tilde{p}_{u S(v)} \big) y_{uS} \, \le 1 &                         & \forall v \in V \label{eqn:reduced-form-lp-online} \\[2ex]
    \min_{\alpha, \beta \ge 0} \quad                                                                                  &
    \sum_{u \in U} \alpha_u + \sum_{v} \beta_v                                                       &                         & \textbf{(Dual)} \notag                                                                                               \\
    \textrm{s.t.} \quad                                                                                        &
    \alpha_u + \sum_{v \in S} \big( 1 - \tilde{p}_{u S(v)} \big) \beta_v \ge \tilde{p}_{uS}          &                         & \forall u \in U, \forall S \subseteq V \label{eqn:reduced-form-dual}                                             
\end{align}

\begin{theorem}
    \label{thm:succint-stochastic-configuration-lp}
    The optimal objective of Reduced-form Stochastic Configuration LP is greater than or equal to $\SOPT$.
\end{theorem}

\begin{proof}
    Consider any computationally unlimited offline algorithm that knows the graph and the edges' success probabilities, but must follow the arrival order of online vertices and the stochastic thresholds of offline vertices.
    It suffices to show that the corresponding $y_{uS}$, i.e., the probability that a subset of online vertices $S$ would be matched to $u$ should $u$ never succeed, is a feasible solution to the Reduced-form Stochastic Configuration LP.

    Constraint~\eqref{eqn:reduced-form-lp-offline} corresponds to the fact that there is a uniquely defined subset of online vertices $S$ that $u$ would be matched to if $u$ never succeed, for any realization of randomness of both the model and the algorithm.

    Next we examine constraint~\eqref{eqn:reduced-form-lp-online}.
    Conditioned on any stochastic thresholds $\threshold_{-u}$ of vertices other than $u$, and the corresponding subset of online vertices $S$ that would be matched to $u$ if $u$ never succeed, an online vertex $v \in S$ would be matched to $u$ if $u$'s stochastic threshold $\threshold_u$ is large enough to accommodate all vertices in $S$ that arrive before $v$, i.e., if $\threshold_u \ge \tilde{p}_{uS(v)}$.
    Since $\threshold_u$ is drawn uniformly between $0$ and $1$, this happens with probability $1 - \tilde{p}_{uS(v)}$.
    Hence, the left-hand-side (LHS) of constraint~\eqref{eqn:reduced-form-lp-online} calculates the probability that $v$ is matched, and is thus at most $1$.
\qed \end{proof}

\begin{lemma}
    \label{lem:succinct-stochastic-configuration-lp-primal-dual}
    Suppose that an algorithm for Online Matching with Stochastic Rewards can further split $p_{uv}$ between $\alpha_u$ and $\beta_v$ each time it matches an edge $(u,v)$, such that the dual constraint~\eqref{eqn:reduced-form-dual} holds up to factor $\Gamma$ \emph{in expectation}, i.e.:
    \[
        \E \Big[ \, \alpha_u + \sum_{v \in S} \big( 1 - \tilde{p}_{u S(v)} \big) \beta_v \, \Big] ~\ge~ \Gamma \cdot \tilde{p}_{uS}
        ~.
    \]
    Then, the algorithm is $\Gamma$-competitive w.r.t.\ $\SOPT$.
\end{lemma}

\begin{proof}
    The lemma's assumptions imply that
    (1) $(\frac{1}{\Gamma} \,\E\, \alpha_u)_{u \in U}$ and $(\frac{1}{\Gamma} \,\E\, \beta_v)_{v \in V}$ form a feasible dual assignment,
    and the corresponding dual objective equals $\frac{1}{\Gamma}$ times the algorithm's expected objective.
    Further, the dual objective of any feasible dual assignment is greater than or equal to the optimal dual objective, which by weak duality is greater than or equal to the optimal primal objective.
    Putting these together with Theorem~\ref{thm:succint-stochastic-configuration-lp} completes the proof.
\qed \end{proof}

\subsection{Stochastic Configuration Linear Program}
Consider $0 \le y_{uS}(\budget_{-u}) \le 1$ (respectively, $y_{uS}(\threshold_{-u})$ if the success probabilities are \emph{not} infinitesimal) as the probability that a subset of online vertices $S$ would be matched to an offline vertex $u$ if it has an infinite stochastic budget $\budget_u = \infty$ (respectively, if $\threshold_u = 1$), \emph{conditioned on the stochastic budgets $\budget_{-u}$} (respectively, the stochastic thresholds $\threshold_{-u}$) of other offline vertices, and over the randomness of the algorithm.
This paper will further consider an even more expressive Stochastic Configuration LP and its dual in some analyses.
\begin{align}
    \maximize_{y \ge 0} \quad &
    \sum_{u \in U} \sum_{S \subseteq V} \, \E_{\budget_{-u}} \Big[ \, \tilde{p}_{uS} \, y_{uS}(\budget_{-u}) \, \Big] && \textbf{(Primal)} \notag \\
    \textrm{subject to} \quad &
    \sum_{S \subseteq V} y_{uS}(\budget_{-u}) \le 1 && \forall u \in U, \forall \budget_{-u} \label{eqn:verbose-form-lp-offline} \\
    &
    \sum_{u} \sum_{S \subseteq V : v \in S, \budget_u \ge p_{uS(v)}} y_{uS}(\budget_{-u}) \le 1 && \forall v \in V, \forall \budget \label{eqn:verbose-form-lp-online} \\[2ex]
    \minimize_{\alpha, \beta \ge 0} \quad &
    \E_{\budget} \Big[ \sum_{u} \alpha_u(\budget_{-u}) + \sum_{v \in V} \beta_v(\budget) \Big] && \textbf{(Dual)} \notag \\
    \textrm{subject to} \quad &
    \alpha_u(\budget_{-u}) + \E_{\budget_{u}} \Big[ \sum_{v \in S : \budget_u \ge p_{uS(v)}} \beta_v(\budget_u, \budget_{-u}) \Big] \ge \tilde{p}_{uS} && \forall u \in U, \forall S \subseteq V, \forall \budget_{-u} \label{eqn:verbose-form-dual}
\end{align}

\begin{theorem}
    \label{thm:stochastic-configuration-lp}
    The optimal objective of Stochastic Configuration LP is at least $\SOPT$.
\end{theorem}

\begin{proof}
    Consider any computationally unlimited offline algorithm that knows the graph and the edges' success probabilities, but must follow the arrival order of online vertices and the stochastic budgets of offline vertices.
    We next show that the corresponding $y_{uS}(\budget_{-u})$, i.e., the probability that a subset of online vertices $S$ would be matched to $u$ if the stochastic budgets were $(\budget_u = \infty, \budget_{-u})$, is a feasible solution to the Stochastic Configuration LP.

    Constraint~\eqref{eqn:verbose-form-lp-offline} corresponds to the fact that there is a uniquely defined subset of online vertices $S$ that $u$ would be matched to if $u$ never succeed, for any realization of randomness of both the model and the algorithm.

    Next we examine constraint~\eqref{eqn:verbose-form-lp-online}.
    Conditioned on any stochastic budgets $\budget_{-u}$ of vertices other than $u$, and the corresponding subset of online vertices $S$ that would be matched to $u$ if $u$ never succeed, an online vertex $v \in S$ would be matched to $u$ if $u$'s stochastic budget $\budget_u$ is large enough to accommodate all vertices $S$ that arrive before $v$, i.e., if we have $\budget_u \ge p_{uS(v)}$.
    Hence, the LHS of constraint~\eqref{eqn:verbose-form-lp-online} calculates the probability that $v$ is matched conditioned on any stochastic budgets $\budget$, and is thus at most $1$.
\qed \end{proof}

\begin{lemma}
    \label{lem:stochastic-configuration-lp-primal-dual}
    Suppose that an algorithm for Online Matching with Stochastic Rewards can further split $p_{uv}$ between $\alpha_u$ and $\beta_v$ each time it matches an edge $(u,v)$.
    Let $\alpha_u(\budget_{-u})$ be the expectation of $\alpha_u$ conditioned on any stochastic budgets $\budget_{-u}$ of offline vertices other than $u$, and let $\beta(\budget)$ be the expectation of $\beta_v$ conditioned on any stochastic budgets $\budget$ of all vertices.
    If the dual constraint~\eqref{eqn:reduced-form-dual} holds up to factor $\Gamma$, i.e.:
    \[
        \alpha_u(\budget_{-u}) + \E_{\budget_{u}} \Big[ \sum_{v \in S : \budget_u \ge p_{uS(v)}} \beta_v(\budget_u, \budget_{-u}) \Big] ~\ge~ \Gamma \cdot \tilde{p}_{uS}
        ~,
    \]
    then the algorithm is $\Gamma$-competitive w.r.t.\ $\SOPT$.
\end{lemma}

\begin{proof}
    The lemma's assumptions imply that
    (1) $\frac{1}{\Gamma} \alpha_u(\budget_{-u})$ and $\frac{1}{\Gamma} \beta_v(\budget)$ are a feasible dual assignment,
    and the corresponding dual objective equals $\frac{1}{\Gamma}$ times the algorithm's expected objective.
    Further, the dual objective of any feasible dual assignment is greater than or equal to the optimal dual objective, which by weak duality is greater than or equal to the optimal primal objective.
    Putting these together with Theorem~\ref{thm:stochastic-configuration-lp} completes the proof.
\qed \end{proof}

\section{Ranking}
\label{sec:ranking}

\subsection{Basics}

We will assume throughout the section that the ranks of offline vertices are distinct, since that happens with probability $1$.
Below we first develop some basic elements that will be used by the proofs of both results in this section.

\smallskip

\noindent
\emph{Dual Updates.~}
We start by explaining the dual update rule associated with Ranking, which is identical to the existing one in the literature.
The dual variables are initially $0$.
Let $g:[0,1] \to [0,1]$ be a non-decreasing gain splitting function to be determined.
We will assume that $g(1) = 1$ which will allow us to handle a boundary case of the analysis under a unified framework.
Recall that $\rank_u$ denotes the rank of $u$ which is uniformly drawn between $0$ and $1$ at the beginning by the Ranking algorithm.
When the algorithm matches an online vertex $v$ to an offline vertex $u$, increase $\alpha_u$ by $p \cdot g(\rank_u)$, and set $\beta_v$ as $p \cdot \big(1-g(\rank_u)\big)$, where recall that $p$ denotes the equal success probability of matching neighboring vertices.

\smallskip
\noindent
\emph{Characterization of Matching.~}
All analyses in this section will fix an offline vertex $u \in U$ and the ranks $\rank_{-u}$ and stochastic thresholds $\threshold_{-u}$ of other offline vertices.
A canonical analysis of Ranking, such as those for Online Bipartite Matching (c.f., Devanur et al.~\cite{DevanurJK:SODA:2013}), would further fix a neighboring online vertex $v \in V$ and characterize the matching outcome of $u$ and $v$.
For the problem of Online Matching with Stochastic Rewards and in the spirit of Configuration LPs, however, we need to characterize the joint matching outcome of $u$ and \emph{all its neighbors}.

We first define some notations.
Consider an imaginary run of Ranking with vertex $u$ removed, while keeping the ranks of other offline vertices as $\rank_{-u}$.
Further consider an online vertex $v$ in $u$'s neighborhood.
If $v$ is matched to an offline vertex $u'$ in the imaginary run, define the $v$'s \emph{critical rank} as $\critical_v = \rank_{u'}$.
If $v$ is not matched, define its critical rank as $\critical_v = 1$.
Finally, for any $0 \le \critical \le 1$, let $N_u(\critical)$ be the set of $u$'s neighbors whose critical ranks are greater than or equal to $\critical$.

\begin{lemma}
    \label{lem:ranking-outcome}
    Suppose that $u$'s rank is $\rank_u$ and $u$'s stochastic threshold is such that $u$ succeeds after $i$ matches.
    Then, the subset of online vertices matched to $u$ is the first $\min \big\{ i, |N_u(\rank_u)| \big\}$ vertices in $N_u(\rank_u)$ according to the arrival order.
\end{lemma}

\begin{proof}
    We will prove a stronger claim consisting of the following three properties, by an induction on $0 \le i \le |N_u(\rank_u)|$:
    \begin{enumerate}
        \item[(a)] the subset of online vertices matched to $u$ is the first $i$ vertices in $N_u(\rank_u)$;
        \item[(b)] $u$'s remaining neighbors in $N_u(\rank_u)$ are either matched to offline vertices with ranks greater than $\rank_u$ or unmatched; and
        \item[(c)] $u$'s neighbors outside $N_u(\rank_u)$ are matched to offline vertices with ranks smaller than $\rank_u$.
    \end{enumerate}

    The base case when $i = 0$ corresponds to running Ranking with $u$ removed from the graph.
    Thus the above properties follow by the definition of $N_u(\rank_u)$.

    Suppose that the above claim holds for some $i \ge 0$.
    Next we consider the case of $i+1$.
    The matchings w.r.t.\ $i$ and $i+1$ merely differ by an alternating path triggered by the $(i+1)$-th match to offline vertex $u$.
    Further by the inductive hypothesis, in particular by (b) and (c), this online vertex is the $(i+1)$-th vertex in $N_u(\rank_u)$.
    Hence, we have established property (a) for the case of $i+1$.
    To see the other two properties, further observe that the ranks of offline vertices on the alternating paths are weakly increasing, because every change is triggered by some online vertex's matching to a weakly better neighbor.
    Therefore, the dichotomy of which $u$'s neighbors are matched to offline vertices with ranks smaller than $\rank_u$ stays intact.

    Finally, we remark that the claim for $i = N_u(\rank_u)$ also proves the lemma for $i > N_u(\rank_u)$.
\qed \end{proof}

\noindent
\emph{Expectation of Dual Variables.~}
The next invariant about dual variable $\alpha_u$ follows from the above gain splitting rule.



\begin{lemma}
    \label{lem:ranking-alpha-bound}
    For any offline vertex $u \in U$, any ranks $\rank$ and stochastic thresholds $\threshold$ of offline vertices, and the corresponding load $\load_u$ of $u$, we have:
    \[
        \alpha_u ~=~ \load_u \cdot g(\rank_u)
        ~.
    \]
\end{lemma}

As a corollary, we can further bound the expectation of $\alpha_u$ conditioned on the ranks $\rank_{-u}$ and stochastic thresholds $\threshold_{-u}$ of other offline vertices.

\begin{lemma}
    \label{lem:ranking-alpha-expectation-bound}
    For any offline vertex $u \in U$, any ranks $\rank_{-u}$ and stochastic thresholds $\threshold_{-u}$ of offline vertices other than $u$, we have:
    \[
        \E_{\rank_u, \threshold_u} \big[ \, \alpha_u \mid \rank_{-u}, \threshold_{-u} \, \big] = \int_0^1 \Big( 1 - (1-p)^{|N_u(\rank_u)|} \Big) g(\rank_u) \D\rank_u
        ~.
    \]
\end{lemma}

\begin{proof}
    Comparing the stated bound with the conclusion of Lemma~\ref{lem:ranking-alpha-bound}, it remains to prove that:
    \begin{equation}
        \label{eqn:ranking-load-expectation-bound}
        \E_{\threshold_u} \big[ \, \load_u \mid \rank, \threshold_{-u} \,\big] = 1 - (1-p)^{|N_u(\rank_u)|}
        ~.
    \end{equation}

    By Lemma~\ref{lem:ranking-outcome}, the subset of online vertices matched to $u$ would be $N_u(\rank_u)$ if $u$ never succeeds.
    Since the expected load of $u$ equals the probability that it succeeds, which is $1 - (1-p)^{|N_u(\rank_u)|}$, we prove Equation~\eqref{eqn:ranking-load-expectation-bound} and thus the lemma.
\qed \end{proof}

Recall that $N_u(\critical)$ is the subset of $u$'s neighbors whose critical ranks are greater than or equal to $\critical$.
Further let $N_u(\critical, v)$ be the subset of vertices in $N_u(\critical)$ that arrive before $v$.

\begin{lemma}
    \label{lem:ranking-beta-expectation-bound}
    For any offline vertex $u \in U$, any $u$'s neighbor $v$, any ranks $\rank_{-u}$ and stochastic thresholds $\threshold_{-u}$ of offline vertices other than $u$, we have:
    \[
        \E_{\rank_u, \threshold_u} \big[ \, \beta_v \mid \rank_{-u}, \threshold_{-u} \, \big] ~\ge~ p \bigg( 1 - g(\critical_v) +\int_0^{\critical_v} \big( 1-p \big)^{|N_u(\rank_u, v)|} \big( g(\critical_v)-g(\rank_u) \big) \D \rank_u \bigg)
        ~.
    \]
\end{lemma}

\begin{proof}
    By the characterization of matching outcome (Lemma~\ref{lem:ranking-outcome}) and the dual update rule, dual variable $\beta_v$ is at least $p \big( 1 - g(\critical_v) \big)$ regardless of what $u$'s rank $\rank_u$ and stochastic threshold $\threshold_u$ are.
    Further, if $u$'s rank is better than $v$'s critical rank, i.e., $\rank_u < \critical_v$, then $v$ would be matched to $u$ if $u$'s stochastic threshold is such that all matches to $u$ before $v$ are unsuccessful.
    If that happens, dual variable $\beta_v$ would increase by $p \big( 1 - g(\critical_v) \big) - p \big( 1 - g(\rank_u) \big) = p \big( g(\rank_u) - g(\critical_v) \big)$ accordingly compared to the aforementioned $p \big( 1 - g(\critical_v) \big)$ baseline.
    By Lemma~\ref{lem:ranking-outcome}, the online vertices matched to $u$ before $v$ are those in $N_u(\rank_u, v)$.
    Therefore the probability that they all fail is $( 1-p )^{|N_u(\rank_u, v)|}$.
    Putting everything together gives the expression on the right-hand-side (RHS).
\qed \end{proof}

\subsection{Non-stochastic Benchmark}

\begin{theorem}
    \label{thm:ranking-non-stocahstic}
    Ranking is at least $0.572$%
    -competitive w.r.t.\ the offline (non-stochastic) benchmark $\OPT$ for any instance with equal success probabilities.
\end{theorem}

The rest of the subsection is devoted to proving Theorem~\ref{thm:ranking-non-stocahstic} by an online primal dual analysis with the Configuration LP.
By Lemma~\ref{lem:configuration-lp-primal-dual}, it suffices to show that $\E \big[ \, \alpha_u + \sum_{v \in S} \beta_v \big] \ge \Gamma \cdot \bar{p}_{uS}$ for any offline vertex $u \in U$ and any subset of online vertices $S \subseteq V$, and with the stated competitive ratio $\Gamma = 0.572$.
We will prove it further conditioning on any ranks $\rank_{-u}$ and stochastic thresholds $\threshold_{-u}$ of other vertices, i.e.:
\begin{equation}
    \label{eqn:ranking-non-stochastic-dual-feasibility}
    \E_{\rank_u, \threshold_u} \bigg[ \, \alpha_u + \sum_{v \in S} \beta_v \mid \rank_{-u}, \threshold_{-u} \, \bigg] ~\ge~ \Gamma \cdot \bar{p}_{uS}
    ~,
\end{equation}

First we apply the bounds for dual variables' expectation from Lemmas~\ref{lem:ranking-alpha-expectation-bound} and \ref{lem:ranking-beta-expectation-bound} to the LHS of the inequality:
\begin{align}
    \E_{\rank_u, \threshold_u} \bigg[ \, \alpha_u + \sum_{v \in S} \beta_v \mid \rank_{-u}, \threshold_{-u} \, \bigg]
    &
    \ge \int_0^1 \big( 1 - \underbrace{\vphantom{\bigg|} (1-p)^{|N_u(\rank_u)|}}_{(a)} \big) g(\rank_u) \D\rank_u
    \label{eqn:ranking-non-stochastic-dual-feasibility-original}
    \\
     & \quad
    + \underbrace{\sum_{v \in S} ~p}_{(b)} \bigg( 1 - g(\critical_v) +\int_0^{\critical_v} \underbrace{\vphantom{\sum_{v \in S}} \big( 1-p \big)^{|N_u(\rank_u, v)|}}_{(c)} \big( g(\critical_v)-g(\rank_u) \big) \D \rank_u \bigg)
    ~.
    \notag
\end{align}

\noindent
\emph{Assumptions.~}
The rest of the analysis will treat the above as a pure inequality problem, where the arrival order and critical ranks of $u$'s neighbors may be chosen arbitrarily by an adversary.
In particular, when we characterize the worst-case scenario of the inequality we shall not concern about how to construct an instance to get the specified critical ranks.
With this treatment, we can make several assumptions that on the one hand are without loss of generality, and on the other hand simplify the subsequent analysis.

First recall the assumption that the offline vertices' ranks are distinct, because the exceptions happen with zero probability.
In other words, we may relax and simplify the RHS assuming distinct ranks, which for example implies that $u$'s rank $\rank_u$ does not equal the critical rank $\critical_v$ of any $v$.
Although the resulting bound may be violated for some zero measure subset of ranks $\rank_u$, the bound still holds after the integration.

Further, it suffices to consider the above inequality for infinitesimal success probabilities.
For $p = 1$ it is the online bipartite matching problem by Karp et al.~\cite{KarpVV:STOC:1990}, for which Ranking is $1-\frac{1}{e}$ competitive.
For any instance with $0 < p < 1$ and any large integer $n$, we can instead consider a new instance such that (i) the equal success probabilities is $p' = 1 - \sqrt[n]{1-p}$, (ii) each online vertex in the original instance has $n$ copies in the new instance, which arrive consecutively and have the same critical ranks as the original vertex, and (iii) the subset $S'$ of concerned in the new instance contains $m = \lfloor \frac{p}{p'} \rfloor$ copies of each vertex $v \in S$ in the original one.
By this construction, the RHS of Eqn.~\eqref{eqn:ranking-non-stochastic-dual-feasibility} is the same in the two instances, up to an error from the rounding of $m$ that diminishes as $n$ tends to infinity.
Further, the RHS of Eqn.~\eqref{eqn:ranking-non-stochastic-dual-feasibility-original} is smaller in new instance because (i) and (ii) ensure that part (a) stays the same and part (c) decreases, and (iii) ensures that the changes to part (b) weakly decreases the expression.
Hence, it suffices to establish the inequality in the new instance which satisfy the assumption of infinitesimal success probability.
We may therefore rewrite the expression as:
\begin{align}
    \E_{\rank_u, \threshold_u} \bigg[ \, \alpha_u + \sum_{v \in S} \beta_v \mid \rank_{-u}, \threshold_{-u} \, \bigg]
    &
    \ge \int_0^1 \big( 1 - e^{-p|N_u(\rank_u)|} \big) g(\rank_u) \D\rank_u + \sum_{v \in S} ~p \big( 1 - g(\critical_v) \big)
    \label{eqn:ranking-non-stochastic-dual-feasibility-small}
    \\
     & \quad
    + \sum_{v \in S} ~ \int_0^{\critical_v} e^{-p|N_u(\rank_u,v)|} \big( 1 - e^{-p} \big) \big( g(\critical_v)-g(\rank_u) \big) \D \rank_u
    ~.
    \notag
\end{align}

Finally, it is sufficient to establish the stated approximate dual feasibility condition when the critical ranks $\critical_v$ are distinct for different $v$.
Otherwise, we may slightly perturb the critical ranks, e.g., by a random noise drawn from $[-\epsilon, \epsilon]$ for a sufficiently small $\epsilon > 0$, to get a new instance satisfying our assumption.
The resulting RHS of Eqn.~\eqref{eqn:ranking-non-stochastic-dual-feasibility-small} converges to the original one when the magnitude of perturbations tends to zero.

For any $0 \le \critical \le 1$, let:
\[
    P(\critical) ~=~ p \cdot \big| N_u(\critical) \cap S \big|
    \quad,\qquad
    \bar{P}(\critical) ~=~ p \cdot \big| N_u(\critical) \setminus S \big|
\]
denote the sums of success probabilities of $u$'s neighbors with critical ranks greater than or equal to $\critical$, inside $S$ and outside $S$ respectively.

\smallskip
\noindent
\emph{Characterization of Worst-case.~}
We next present a series of lemmas that characterize the worst-case of various aspects of the instance, including the arrival order, the size of $S$, and critical ranks, etc.
We defer the relative long proofs to Appendix~\ref{app:ranking} for readability.

\begin{lemma}[Worst Arrival Order]
    \label{lem:ranking-non-stochastic-dual-feasibility-worst-order}
    We have:
    \begin{align}
        \E_{\rank_u, \threshold_u} \bigg[ \, \alpha_u + \sum_{v \in S} \beta_v \mid \rank_{-u}, \threshold_{-u} \, \bigg]
         &
        ~\ge~
        \int_0^1 \Big( 1 - e^{-\bar{P}(\rank_u)} \Big) g(\rank_u) \D\rank_u
        - \int_0^1 \big( 1 - g(\critical_v) \big) \D P(\critical_v)
        \notag    \\
         & \qquad
        - \int_0^1 e^{-P(\rank_u)-\bar{P}(\rank_u)} \int_{\rank_u}^1 g(\critical_v) \D e^{P(\critical_v)} \D \rank_u
        ~.
        \label{eqn:ranking-non-stochastic-dual-feasibility-worst-order}
    \end{align}
\end{lemma}

The lemma's inequality would hold with equality if online vertices in $S$ arrive by increasing order of critical ranks, and vertices outside $S$ arrive before those in $S$.
We remark that the negative signs in front of the second and third integration come from the fact that $P(\critical)$ is non-increasing.
All three parts make positive contribution to the RHS.
The proof is in Appendix~\ref{app:ranking-non-stochastic-dual-feasibility-worst-order}.

\begin{lemma}[Worst Critical Ranks Outside $S$]
    \label{lem:ranking-non-stochastic-dual-feasibility-worst-outside}
    Given any non-increasing function $P : [0, 1] \to [0, \infty)$, the RHS of Eqn.~\eqref{eqn:ranking-non-stochastic-dual-feasibility-worst-order} is minimized when $\bar{P}(\critical)$ is a step function that equals $\infty$ for $\critical < \critical_0$ and equals $0$ for $\critical > \critical_0$, for some threshold $\critical_0$.
\end{lemma}

The lemma indicates that letting all online vertices outside $S$ have the same critical rank $\critical_0$ is the worst-case scenario.

\begin{proof}
    For any $0 \le \rank_u \le 1$, gather the first and third terms which involve $e^{-\bar{P}(\rank_u)}$ on the RHS of Equation~\eqref{eqn:ranking-non-stochastic-dual-feasibility-worst-order}.
    We get:
    \[
        \Big( - g(\rank_u) - e^{-P(\rank_u)}\int_{\rank_u}^1 g(\critical_v) d e^{P(\critical_v)} \Big) e^{-\bar{P}(\rank_u)} = \Big( \underbrace{-g(1) + \int_{\rank_u}^1 e^{P(\critical_v)} dg(\critical_v)}_{(\star)} \Big) e^{-P(\rank_u)} e^{-\bar{P}(\rank_u)}
        ~,
    \]
    where equality follows from an integration by parts.
    Since $(\star)$ is decreasing in $\rank_u$, the lemma follows with $\critical_0$ being the supremum of $\rank_u$ for which $(\star)$ is non-negative.
    In the boundary case when $(\star)$ is negative for all $\rank_u$, the lemma still holds with $\critical_0 = 0$.
    %
    %
\qed \end{proof}


Given Lemma~\ref{lem:ranking-non-stochastic-dual-feasibility-worst-outside}, and noting that $p_{uS} = P(0)$ and thus $\bar{p}_{uS} = \min \big\{ P(0), 1 \big\}$, it suffices to find a gain splitting function $g$ such that for any non-increasing function $P : [0, 1] \to [0, \infty)$ we have:
\begin{equation}
    \label{eqn:ranking-non-stochastic-dual-feasibility-worst-outside}
    \int_0^{\critical_0} g(\rank_u) \D\rank_u
    - \int_0^1 \big( 1 - g(\critical_v) \big) \D P(\critical_v)
    - \int_{\critical_0}^1 e^{-P(\rank_u)} \int_{\rank_u}^1 g(\critical_v) \D e^{P(\critical_v)} \D \rank_u
    ~\ge~
    \Gamma \cdot \min \big\{ P(0), 1 \big\}
    ~.
\end{equation}

\begin{lemma}[Worst Size of $S$]
    \label{lem:ranking-non-stochastic-dual-feasibility-worst-size}
    If Eqn.~\eqref{eqn:ranking-non-stochastic-dual-feasibility-worst-outside} holds for all non-increasing function $P$ that satisfies $p_{uS} = P(0) = 1$, then it also holds for an arbitrary non-increasing function $P$.
\end{lemma}

\begin{proof}
    For this proof, it is more convenient to write the third term on the left as:
    \[
        - \int_{\critical_0}^1 \int_{\rank_u}^1 e^{-P(\rank_u) + P(\critical_v)} g(\critical_v) \D P(\critical_v) \D \rank_u
        ~.
    \]

    Given any instance whose function $P$ has $P(0) > 1$, we may consider another instance defined by $P'(\critical) = \big( P(\critical) - P(0) + 1 \big)^+$, i.e., dropping some neighbors with the smallest critical ranks from $S$ so that the sum of remaining vertices' success probabilities equals $1$.
    The inequality's RHS equals $\Gamma$ in both instances.
    The LHS, however, is smaller in the new instance.
    To see this, the first term on the left stays the same.
    The second term decreases.
    The third term also decreases because for any relevant $\rank_u \le \critical_v$ with $P'(\rank_u) \ge P'(\critical_v) > 0$ that contribute the integral of the new instance, we have $e^{-P(\rank_u) + P(\critical_v)} = e^{-P'(\rank_u) + P'(\critical_v)}$ because the values of $P$ and $P'$ differ by the same amount $P(0) - 1$ at both $\rank_u$ and $\critical_v$.

    Next consider any instance whose function $P$ has $P(0) < 1$.
    We may consider another instance defined by $P'(\critical) = \nicefrac{P(\critical)}{P(0)}$, i.e., increasing the number of vertices in $S$ with each critical rank by a $\nicefrac{1}{P(0)}$ factor so that the sum of success probabilities becomes $1$.
    By definition, the inequality's RHS increases by a $\nicefrac{1}{P(0)}$ factor in the new instance.
    The LHS, however, increases by less than a $\nicefrac{1}{P(0)}$ factor.
    To see this, the first term on the left remains the same.
    The second term increases by a $\nicefrac{1}{P(0)}$ factor.
    Finally, the third term increases by less than a $\nicefrac{1}{P(0)}$ factor because we have $e^{-P'(\rank_u) + P'(\critical_v)} < e^{-P(\rank_u) + P(\critical_v)}$.
    %
    %
\qed \end{proof}

Therefore, our task further simplifies to finding a gain splitting function $g$ such that for any non-increasing function $P : [0, 1] \to [0, 1]$ with $P(0) = 1$ we have:
\begin{equation}
    \label{eqn:ranking-non-stochastic-dual-feasibility-worst-size}
    \int_0^{\critical_0} g(\rank_u) \D\rank_u
    - \int_0^1 \big( 1 - g(\critical_v) \big) \D P(\critical_v)
    - \int_{\critical_0}^1 e^{-P(\rank_u)} \int_{\rank_u}^1 g(\critical_v) \D e^{P(\critical_v)} \D \rank_u
    ~\ge~
    \Gamma
    ~.
\end{equation}

This is the moment when we finally specify the gain splitting function $g$:
\begin{equation}
    \label{eqn:ranking-non-stochastic-gain-function}
    g(\rank) =
    \begin{cases}
        \min \big\{ \frac{c}{e-(e-1)\rank}, 1-\frac{1}{e} \big\} & 0 \le \rank < 1; \\
        1                                                    & \rank = 1.
    \end{cases}
\end{equation}
for a constant $c \approx 1.161$ such that $\int_0^1 g(x)=1-g(0) > 0.572$.

\begin{lemma}[Worst Critical Ranks Inside $S$]
    \label{lem:ranking-non-stochastic-dual-feasibility-worst-inside}
    For the function $g$ in Eqn.~\eqref{eqn:ranking-non-stochastic-dual-feasibility-worst-size} is minimized when $P$ approaches a step function in the limit, i.e., when for some $\rank_0 > \critical_0$:
    \[
        P(\rank) =
        \begin{cases}
            1                              & 0 \le \rank < \rank_0 - \epsilon;         \\
            \frac{\rank_0-\rank}{\epsilon} & \rank_0 - \epsilon \le \rank \le \rank_0; \\
            0                              & \rank_0 < \rank \le 1.
        \end{cases}
    \]
    and let $\epsilon \to 0$.
\end{lemma}

Despite its complex look, the lemma actually gives a simple characterization of the worst-case critical ranks inside $S$:
all online vertices in $S$ have the same critical rank $\rank_0>\critical_0$.
We have the complex form because we write the lower bounds of the LHS of approximate dual feasibility through a function $P$, which help simplify the proofs of previous lemmas but is inconvenient when we need to represent identical critical ranks.
The proof of Lemma~\ref{lem:ranking-non-stochastic-dual-feasibility-worst-inside} is in Appendix~\ref{app:ranking-non-stochastic-dual-feasibility-worst-inside}.

Finally, we apply the above worst-case function $P$ and focus on optimizing the gain splitting function $g$ w.r.t.\ the resulting differential inequality.

\begin{lemma}
    For any $\critical_0 < \rank_0$, and for the function $g$ in Eqn.~\eqref{eqn:ranking-non-stochastic-gain-function}, we have:
    \[
        \int_0^{\critical_0} g(\rank_u) \D \rank_u + \Big( 1 - g(\rank_0) \Big) + \Big(1-\frac{1}{e}\Big) \big(\rank_0 - \critical_0\big) g(\rank_0) ~\ge~ \Gamma
        ~,
    \]
    for the stated competitive ratio $\Gamma = 0.572$
\end{lemma}

\begin{proof}
    Recall that $1-g(0) = \int_0^1 g(\rank_u) \D \rank_u > \Gamma$.
    The derivative of LHS w.r.t.\ $\critical_0$ is:
    \[
        g(\critical_0)- \Big(1-\frac{1}{e}\Big)g(\rank_0)
    \]

    If $\rank_0 < 1$, we have:
    \[
        g(\critical_0)-\Big(1-\frac{1}{e}\Big)g(\rank_0) \geq g(0)-\Big(1-\frac{1}{e}\Big)^2 > 0
    \]

    Therefore taking $\critical_0 = 0$ minimizes the LHS, with minimum value:
    \begin{align*}
        1-g(\rank_0)+\Big(1-\frac{1}{e}\Big) \cdot \rank_0 \cdot g(\rank_0)
        &
        = 1 - \Big(1 - \Big(1-\frac{1}{e}\Big) \rank_0 \Big) g(\rank_0) \\
        &
        \geq 1 - \Big(1 - \Big(1-\frac{1}{e}\Big) \rank_0 \Big) \frac{c}{e-(e-1)\rank_0} && \mbox{(definition of $g$)} \\
        &
        = 1-\frac{c}{e} = 1-g(0) > \Gamma
        ~.
    \end{align*}

    If $\rank_0=1$, on the other hand, the LHS is:
    \[
        \int_0^{\critical_0} g(\rank_u) \D \rank_u + \Big(1-\frac{1}{e}\Big) \big(1 - \critical_0\big)
        ~.
    \]

    The derivative w.r.t.\ $\critical_0$ is $g(\critical_0) - \big(1 - \frac{1}{e}\big) \le 0$.
    Hence, the LHS is minimized when $\critical_0 = 1$, with minimum value $\int_0^1 g(\rank_u) \D \rank_u > \Gamma$.
\qed \end{proof}


\subsection{Stochastic Benchmark}

\begin{theorem}
    \label{thm:ranking-stochastic}
    Ranking is $1-\frac{1}{e}$-competitive w.r.t.\ the offline stochastic benchmark $\SOPT$ for any instance with equal success probabilities.
\end{theorem}

The rest of the subsection is devoted to proving Theorem~\ref{thm:ranking-stochastic} by an online primal dual analysis with the Reduced-form Stochastic Configuration LP.
By Lemma~\ref{lem:succinct-stochastic-configuration-lp-primal-dual}, it is sufficient to show that $\E\big[\alpha_u + \sum_{v\in S} (1-\tilde{p}_{uS(v)})\beta_v\big]\geq \Gamma \,\tilde{p}_{uS}$
for any offline vertex $u \in U$ and any subset of online vertices $S \subseteq V$, and with the stated competitive ratio $\Gamma = 1-\frac{1}{e}$.
We will prove it further conditioning on any ranks $\rank_{-u}$ and stochastic thresholds $\threshold_{-u}$ of other vertices, i.e.:

\begin{equation}
    \label{eqn:ranking-stochastic-original-condition}
    \E_{\rank_u, \threshold_u} \bigg[ \, \alpha_u + \sum_{v \in S} (1-\tilde{p}_{uS(v)})\beta_v \mid \rank_{-u}, \threshold_{-u} \, \bigg]\geq \Big( 1-\frac{1}{e} \Big) \cdot \tilde{p}_{uS}
\end{equation}

Let $v_1, v_2, \dots, v_n$ be the online vertices in $S$, sorted by their critical ranks, which we denote as $\critical_1 \ge \critical_2 \ge \dots \ge \critical_n$.
We next apply the lower bounds for the dual variables' expectation to the LHS of Eqn.~\eqref{eqn:ranking-stochastic-original-condition}.
First we have:
\begin{align*}
    \E_{\rank_u,\threshold_u} \big[ \, \alpha_u \mid \rank_{-u},\threshold_{-u}\, \big]
    &
    \ge \int_0^1 \Big( 1 - (1-p)^{|N_u(\rank_u)|} \Big) g(\rank_u) \D\rank_u \qquad\qquad\qquad\qquad\qquad \mbox{(Lemma~\ref{lem:ranking-alpha-expectation-bound})} \\[1ex]
    &
    \ge \int_0^1 \Big( 1 - (1-p)^{|N_u(\rank_u) \cap S|} \Big) g(\rank_u) \D\rank_u \\
    &
    = \int_0^1 \sum_{i=1}^{|N_u(\rank_u) \cap S|} p (1-p)^{i-1} g(\rank_u) \D\rank_u 
    = \sum_{i=1}^n p (1-p)^{i-1} \int_0^{\critical_i} g(\rank_u) \D\rank_u
    ~.
\end{align*}

To bound the expected contribution by $\beta_v$'s, we apply a weaker version of Lemma~\ref{lem:ranking-beta-expectation-bound}, dropping the second part on the RHS of the lemma's inequality.
Suppose that the vertices in $S$ arrive by order $v_{\pi(1)} \before v_{\pi(1)} \before \dots \before v_{pi(n)}$.
We get that:
\begin{align*}
    \E_{\rank_u, \threshold_u} \bigg[ \,\sum_{v \in S} (1-\tilde{p}_{uS(v)})\beta_v \mid \rank_{-u}, \threshold_{-u} \, \bigg]
    &
    \ge \sum_{i=1}^{n} p(1-p)^{i-1} \big(1-g(\critical_{\pi(i)})\big) 
    && \mbox{(Lemma~\ref{lem:ranking-beta-expectation-bound})} \\
    &
    \ge \sum_{i=1}^{n} p(1-p)^{i-1} \big(1-g(\critical_i)\big)
    ~.
    && \mbox{(rearrangement inequality)}
\end{align*}

Combining the two bounds, we get that the LHS of Eqn.~\eqref{eqn:ranking-stochastic-original-condition} is at least:
\[
    \sum_{i=1}^n p(1-p)^{i-1} \Big( \underbrace{\int_0^{\critical_i} g(\rank_u) \D\rank_u + 1 - g(\critical_i)}_{(\star)} \Big)
\]

Next we choose $g(x)=e^{x-1}$ just like the analysis of Ranking for the original Online Bipartite Matching problem.
This choice ensures that $(\star)$ equals $1-\frac{1}{e}$.
The above bound is therefore:
\[
    \Big( 1-\frac{1}{e} \Big) \sum_{i=1}^n p (1-p)^{i-1} = \Big( 1-\frac{1}{e} \Big)\Big(1 - (1-p)^n\Big) = \Big( 1-\frac{1}{e} \Big) \cdot \tilde{p}_{uS}
    ~.
\]


%
%
%
\bibliographystyle{splncs04}
\bibliography{stochastic-rewards}

\newpage

\appendix
\section{Stochastic Balance}
\label{sec:balance}

\subsection{Basics}

This section covers the online primal dual analysis of the Stochastic Balance algorithm using the Stochastic Configuration LP.
The success probabilities are infinitesimal throughout this section.
We first develop some basic elements that will be used by the proofs of both results in the section.

\smallskip
\noindent
\emph{Dual Updates.~}
Let $g : [0, \infty) \to [0, 1]$ be a non-decreasing gain splitting function to be determined.
Further let $G(\load) = \int_0^\load g(z) \D z$ be the integral of $g$.
When the algorithm matches $u$ and $v$ by an $x_{uv}$ amount, the algorithm adds $x_{uv} p_{uv}$ to the objective in expectation.
Accordingly, it increases $\alpha_u(\budget_{-u})$ by $x_{uv} p_{uv} g(\load_u)$ and increases $\beta_v(\budget)$ by $x_{uv} p_{uv} \big( 1 - g(\load_u) \big)$, where $\budget$ is the  stochastic budgets of offline vertices.
Here recall that $\load_u$ denotes the (current) load of an offline vertex $u$.


\bigskip

%
The following invariant follows by how the algorithm updates dual variable $\alpha_u(\budget_{-u})$.

\begin{lemma}
    \label{lem:alpha-invariant}
    For any offline vertex $u \in U$ and conditioned on any stochastic budgets $\budget_{-u}$ of vertices other than $u$, we have $\alpha_u(\budget_{-u}) = G(\load_u)$.
\end{lemma}

Let $\load^\infty_u$ denote what $u$'s final load would be if all matches to it were unsuccessful, i.e., if its stochastic budget was $\budget_u = \infty$.
A simple calculation shows the following:

\begin{lemma}
    \label{lem:balance-alpha-bound}
    For any offline vertex $u \in U$ and conditioned on any stochastic budgets $\budget_{-u}$ of vertices other than $u$, and for the corresponding $\load^\infty_u$, we have:
    \[
        \E_{\budget_u} \, \alpha(\budget_{-u}) = \int_0^{\load^\infty_u} e^{-\budget_u} g(\budget_u) \D \budget_u
        ~.
    \]
\end{lemma}

\begin{proof}
    Offline vertex $u$'s final load is $\min \big\{ \load^\infty_u, \budget_u \big\}$.
    Hence by Lemma~\ref{lem:alpha-invariant} we get that:
    \begin{align*}
        \E_{\budget_u} \, \alpha(\budget_{-u})
         &
        = \int_0^\infty e^{-\budget_u} G\big( \min \big\{ \load^\infty_u, \budget_u \big\} \big) \D \budget_u
         &   &
        \mbox{(substitution of variable)}                                                                                                               \\
         &
        = \int_0^{\load^\infty_u} e^{-\budget_u} G(\budget_u) \D \budget_u + \int_{\load^\infty_u}^\infty e^{-\budget_u} G(\load^\infty_u) \D \budget_u \\
         &
        = \int_0^{\load^\infty_u} e^{-\budget_u} G(\budget_u) \D \budget_u + e^{-\load^\infty_u} G\left( \load^\infty_u \right)                         \\
         &
        = \int_0^{\load^\infty_u} e^{-\budget_u} g(\budget_u) \D \budget_u
        ~.
         &   &
        \mbox{(integration by parts)}
    \end{align*}
\qed \end{proof}

%
%
%

%

\subsection{Equal Infinitesimal Success Probabilities}

\begin{theorem}
    \label{thm:balance-equal-prob}
    Stochastic Balance is at least $2(1-\ln 2) > 0.613$-competitive w.r.t.\ the offline stochastic benchmark $\SOPT$, when the instance has equal infinitesimal success probabilities.
\end{theorem}

The rest of the subsection is devoted to proving Theorem~\ref{thm:balance-equal-prob}.
By Lemma~\ref{lem:stochastic-configuration-lp-primal-dual}, we need to show that for any offline vertex $u \in U$ and any subset of online vertices $S \subseteq V$, and conditioned on any stochastic budgets $\budget_{-u}$ of other offline vertices, we have:
\[
    \E_{\budget_u} \Big[ \, \alpha_u(\budget_{-u}) + \sum_{v \in S : \budget_u \ge p_{uS(v)}} \beta_v(\budget_u, \budget_{-u}) \, \Big] \ge \Gamma \cdot \tilde{p}_{uS}
    ~,
\]
for the stated competitive ratio $\Gamma = 2(1-\ln 2)$.
Note that it is without loss of generality to consider $S$ that is comprised of $u$'s neighbors.

We start by bounding the expectation of dual variables while leaving the gain splitting function $g$ to be determined.
The expectation of $\alpha(\budget_{-u})$ is already given by Lemma~\ref{lem:balance-alpha-bound}.
The next lemma bounds the expectation of $\beta_v(\budget_u, \budget_{-u})$'s.

\begin{lemma}
    \label{lem:balance-equal-prob-dual-gain}
    Conditioned on any $\budget_{-u}$ and the corresponding $\load^\infty_u$, we have:
    \[
        \E_{\budget_u} \Big[ \,\sum_{v \in S : \budget_u \ge p_{uS(v)}} \beta_v(\budget_u, \budget_{-u}) \, \Big]
        \ge \int_0^\infty e^{-\budget_u} \Big( \min \big\{ p_{uS}, \budget_u \big\} - (\load^\infty_u - \budget_u)^+ \Big)^+ \D \budget_u \cdot \Big( 1 - g(\load^\infty_u) \Big)
        ~.
    \]
\end{lemma}

\begin{proof}

    If $u$'s stochastic budget is at least $\budget_u \ge \load^\infty_u$, then $u$ remains unsuccessful until the end and therefore any vertex $v \in S$ would be matched to an offline vertex with load at most $\load^\infty_u$.
    Further note that $\sum_{v \in S : \budget_u \ge p_{uS(v)}} p_{uv} = \min \big\{ p_{uS}, \budget_u \big\}$.
    Hence:
    \[
        \sum_{v \in S : \budget_u \ge p_{uS(v)}} \beta_v(\budget_u, \budget_{-u}) \ge \min \big\{ p_{uS}, \budget_u \big\} \Big( 1 - g(\load^\infty_u) \Big)
    \]

    Next consider the case when $u$'s stochastic budget is $\budget_u < \load^\infty_u$.
    First consider how much the $\beta_v$'s would sum to if $u$'s stochastic budget was $\infty$ (while using the actual $\budget_u$ to define the range of summation).
    By the same argument as above, all vertices in the summation would be matched to an offline vertex with load at most $\load^\infty_u$.
    We can therefore lower bound the actual value of $\beta_v$'s using the following structural lemma.

    \begin{lemma}[c.f., Lemma 7 of Huang and Zhang~\cite{HuangZ:STOC:2020}]
        If we compare the set of online vertices that are matched to offline vertices with loads at most $\load^\infty_u$ when $u$'s stochastic budget is $\budget_u$, with the set of such vertices when $u$'s stochastic budget is $\infty$, they differ by at most $p^{-1} \big( \load^\infty_u - \budget_u \big)^+$ vertices.
    \end{lemma}

    We get that:
    \[
        \sum_{v \in S : \budget_u \ge p_{uS(v)}} \beta_v(\budget_u, \budget_{-u}) \ge \Big( \min \big\{ p_{uS}, \budget_u \big\} - \big( \load^\infty_u - \budget_u\big) \Big)^+ \Big( 1 - g(\load^\infty_u) \Big)
        ~.
    \]

    Combining the two cases proves the lemma.
\qed \end{proof}

Next we apply the bounds of dual variables and focus on optimizing the gain splitting function $g$.
For ease of notations, we will now omit the superscripts and subscripts, replacing $\load^\infty_u$ with $\load$, $\budget_u$ with $\budget$, and $p_{uS}$ with $q$ (because $p$ is already devoted to denote the equal success probabilities).
To finish the proof of Theorem~\ref{thm:balance-equal-prob}, it suffices to show the next lemma.

\begin{lemma}
    There is a non-decreasing function $g : [0, \infty) \to [0, 1]$ such that for any $\load \ge 0, q \ge 0$:
    \[
        \int_0^\load e^{-\budget} g(\budget) \D \budget + \int_0^\infty e^{-\budget} \big( \min \{ q, \budget \} - (\load - \budget)^+ \big)^+ \D \budget \cdot \big(1-g(\load)\big) \ge \Gamma \cdot (1 - e^{-q})
        ~,
    \]
    for $\Gamma = 2(1-\ln 2) > 0.613$.
\end{lemma}

\begin{proof}
    First consider the boundary condition when $\load = 0$:
    \[
        \int_0^\infty e^{-\budget} \min \{q, \budget\} \D \budget \cdot \big(1-g(0)\big) = \big( 1-e^{-q} \big) \big( 1-g(0) \big) \ge \Gamma \big(1 - e^{-q} \big)
        ~.
    \]

    Hence, we let $g(0) = 1 - \Gamma$.

    Next we simplify the differential inequality by characterizing the worst-case value of $q$ for any given $\load$.
    Consider increasing $q$ by an infinitesimal amount $\epsilon$.
    How would this affect the LHS?
    The first integral stays the same.
    In the second integral, the term $\big( \min \{ p_{uS}, \budget_u \} - ( \load^\infty - \budget_u) \big)^+$ increases by at most $\epsilon$ for $\budget > q$.
    Hence, the derivative of LHS w.r.t.\ $q$ is at most:
    \[
        \int_q^\infty e^{-\budget} \D \budget \cdot \big(1-g(\load)\big)
        = e^{-q} \big(1-g(\load)\big)
        ~.
    \]

    On the other hand, the derivative of RHS w.r.t.\ $q$ is $e^{-q}\,\Gamma$.
    By the monotonicity of $g$ and the boundary condition that $g(0) = 1-\Gamma$, this is $e^{-q} \big(1-g(0)\big) \ge e^{-q} \big(1-g(\load)\big)$.
    In other words, the derivative of RHS is greater than or equal to the derivative of LHS.

    Therefore, the lemma's differential inequality holds if and only if it holds for $q = \infty$.
    Simplifying the LHS with basic calculus after plugging in $q = \infty$, our task simplifies to finding a function $g$ such that for any $\load \ge 0$:
    \[
        \int_0^\load e^{-\budget} g(\budget) \D \budget + \big( 2 e^{-\load / 2} - e^{-\load} \big) \big(1-g(\load)\big) \ge \Gamma
        ~.
    \]

    To further reduce the problem to solving an ordinary differential equation, we next substitute variables by letting $\mu = e^{-\budget}$ and $f(\mu) = g(- \ln \mu)$.
    The problem then becomes finding a function $f$ such that for any $0 \le \lambda \le 1$:
    \[
        \int_\lambda^1 f(\mu) \D \mu + \big( 2 \sqrt{\lambda} - \lambda \big) \big(1 - f(\lambda) \big) \ge \Gamma
        ~,
    \]
    with boundary conditions $f(1) = 1 - \Gamma$ and $f(0) \le 1$.
    The lemma's competitive ratio $\Gamma = 2(1-\ln 2)$ is achieved when this holds with equality for all $0 \le \lambda \le 1$, which reduces to an ordinary differential equation.
    Solving it gives:
    \[
        f(\mu) = \begin{cases}
            \frac{2 (2 - \sqrt{\mu}) (\ln 2 - \ln (2-\sqrt{\mu}))}{\sqrt{\mu}} - 1 & \mbox{if $\mu > 0$;} \\
            1                                                                      & \mbox{if $\mu = 0$.}
        \end{cases}
    \]
\qed \end{proof}

\subsection{General Infinitesimal Success Probabilities}

\begin{theorem}
    \label{thm:balance-general}
    Stochastic Balance is at least $0.611$-competitive w.r.t.\ the offline stochastic benchmark $\SOPT$, when the instance has (general) infinitesimal success probabilities.
\end{theorem}

The rest of the subsection is devoted to proving Theorem~\ref{thm:balance-general}.
By Lemma~\ref{lem:stochastic-configuration-lp-primal-dual}, we need to show that for any offline vertex $u \in U$, any subset of online vertices $S \subseteq V$, and any stochastic budgets $\budget_{u}$ of other offline vertices, we have:
\[
    \E_{\budget_u} \Big[ \, \alpha_u(\budget_{-u}) + \sum_{v \in S : \budget_u \ge p_{uS(v)}} \beta_v(\threshold) \, \Big] \ge \Gamma \cdot \tilde{p}_{uS}
    ~,
\]
for the stated competitive ratio $\Gamma = 0.611$.

Once again, we start by bounding the expectation of dual variables while leaving the gain splitting function $g$ to be determined.
Recall that $\load^\infty_u$ denotes the load of $u$ if all matches to $u$ fail, i.e., if its stochastic budget was $\budget_u = \infty$.
The expectation of $\alpha(\budget_{-u})$ is already given by Lemma~\ref{lem:balance-alpha-bound}.
The next lemma bounds the expectation of $\beta_v(\threshold)$'s.

\begin{lemma}
    \label{lem:balance-general-dual-gain}
    Conditioned on any $\budget_{-u}$ and the corresponding $\load^\infty_u$, for any $\budget_u$, we have:
    \[
        \sum_{v \in S : \budget_u \ge p_{uS(v)}} \beta_v(\budget_u, \budget_{-u})
        \ge
        \begin{cases}
            \min \big\{ p_{uS}, \budget_u \big\} \cdot \big( 1 - g(\load^\infty_u) \big)
             &
            \mbox{if $\budget_u \ge \load^\infty_u$;}
            \\[2ex]
            \Big( \min \big\{ p_{uS}, \budget_u \big\} \cdot \big( 1 - g(\load^\infty_u) \big) - \int_{\budget_u}^{\load^\infty_u} \big( 1 - g(z) \big) dz \Big)^+
             &
            \mbox{if $\budget_u > \load^\infty_u$.}
        \end{cases}
    \]
\end{lemma}


\begin{proof}
    If $\budget_u \ge \load^\infty_u$, the proof is verbatim to the counterpart in Lemma~\ref{lem:balance-equal-prob-dual-gain}.
    Since $u$ remains unsuccessful until the end,  any neighboring online vertex $v \in S$ would be matched to an offline vertex with load at most $\load^\infty_u$.
    Further note that $\sum_{v \in S : \budget_u \ge p_{uS(v)}} p_{uv} = \min \big\{ p_{uS}, \budget_u \big\}$.
    Hence:
    \[
        \sum_{v \in S : \threshold_u \ge \tilde{p}_{uS(v)}} \beta_v(\threshold) \ge \min \big\{ p_{uS}, \budget_u \big\} \Big( 1 - g(\load^\infty_u) \Big)
        ~.
    \]

    Next consider the case when $\budget_u < \load^\infty_u$.
    The proof is also similar to the counterpart in Lemma~\ref{lem:balance-equal-prob-dual-gain}, except that we would need a different structural lemma.
    First consider how much the $\beta_v$'s would sum to if $u$'s stochastic budget was $\infty$ (while using the actual $\budget_u$ to define the range of summation).
    By the same argument as above, we have that all vertices in the summation would be matched to an offline vertex with load at most $\load^\infty_u$.
    In other words:
    \[
        \sum_{v \in S : \budget_u \ge p_{uS(v)}} \beta_v(\budget'_u = \infty, \budget_{-u}) \ge \min \big\{ p_{uS}, \budget_u \big\} \Big( 1 - g(\load^\infty_u) \Big)
        ~.
    \]

    We can then lower bound the actual value of $\beta_v$'s using the following structural lemma.

    \begin{lemma}[c.f., Lemma 22 of Huang and Zhang~\cite{HuangZ:STOC:2020}]
        For any subset of online vertices $T \subseteq V$, if we compare the value of $\sum_{v \in T} \beta_v$ when $u$'s stochastic budget is $\budget_u$, with the counterpart when $u$'s stochastic budget is $\infty$, they differ by at most $\int_{\budget_u}^{\load^\infty_u} \big( 1 - g(z) \big) dz$.
    \end{lemma}

    On the one hand, by the above argument we get that:
    \[
        \sum_{v \in S : \budget_u \ge p_{uS(v)}} \beta_v(\budget_u, \budget_{-u}) \ge \min \big\{ p_{uS}, \budget_u \big\} \Big( 1 - g(\load^\infty_u) \Big) - \int_{\budget_u}^{\load^\infty_u} \big( 1 - g(z) \big) dz
        ~.
    \]

    On the other hand, the LHS above is nonnegative.
    The lemma follows.
\qed \end{proof}

Our last task is to apply the bounds of dual variables and focus on optimizing the gain splitting function $g$.
Once again, we shall now omit the superscripts and subscripts  for ease of notations, replacing $\load^\infty_u$ with $\load$, $\budget_u$ with $\budget$, and $p_{uS}$ with $q$ (because $p$ is already devoted to denote the equal success probabilities).
It suffices to prove the next lemma.

\begin{lemma}
    There is a non-decreasing function $g : [0, \infty) \to [0, 1]$ such that for any $\load \ge 0, q \ge 0$:
    \begin{multline*}
        \int_0^\load e^{-\budget} g(\budget) d\budget + \int_0^\load e^{-\budget} \Big( \min \{ q, \budget \} \big(1-g(\load)\big)  - \int_{\budget}^{\load} \big( 1 - g(z) \big) dz \Big)^+ d\budget \\
        + \int_\load^\infty e^{-\budget}  \min \{ q, \budget \} \big(1-g(\load)\big) d\budget\ge \Gamma \cdot (1 - e^{-q})
        ~,
    \end{multline*}
    for $\Gamma = 0.611$.
\end{lemma}

\begin{proof}
    First consider the boundary condition when $\load = 0$:
    \[
        \int_0^\infty e^{-\budget} \min \{q, \budget\} d\budget \big(1-g(0)\big) = \big( 1-e^{-q} \big) \big( 1-g(0) \big) \ge \Gamma \big(1 - e^{-q} \big)
        ~.
    \]

    Hence, we let $g(0) = 1 - \Gamma$.

    Next we simplify the differential inequality by characterizing the worst-case value of $q$ for any given $\load$.
    First rewrite the relevant terms (the second and the third) on the LHS that involves $q$ as:
    \[
        \int_0^\infty e^{-\budget} \underbrace{\Big( \min \{ q, \budget \} \big(1-g(\load)\big)  - \int_{\min\{\budget, \load\}}^{\load} \big( 1 - g(z) \big) dz \Big)^+}_{(\star)} d\budget
        ~.
    \]
    Consider increasing $q$ by an infinitesimal amount $\epsilon$.
    This will increase $(\star)$ by at most $\epsilon$ for $\budget > q$.
    Hence, the derivative of LHS w.r.t.\ $q$ is at most:
    \[
        \int_q^\infty e^{-\budget} d\budget \big(1-g(\load)\big)
        = e^{-q} \big(1-g(\load)\big)
        ~.
    \]

    On the other hand, the derivative of RHS w.r.t.\ $q$ is $e^{-q} \, \Gamma$.
    By the monotonicity of $g$ and the boundary condition that $g(0) = 1 - \Gamma$, this is $e^{-q} \big(1-g(0)\big) \ge e^{-q} \big(1-g(\load)\big)$.
    In other words, the derivative of RHS is greater than or equal to the derivative of LHS.

    Therefore, the lemma's differential inequality holds for any value if and only if it holds for $q = \infty$.
    Simplifying the LHS with basic calculus after plugging in $q = \infty$, our task simplifies to finding a function $g$ such that for any $\load \ge 0$:
    \[
        \int_0^\load e^{-\budget} g(\budget) d\budget + \int_0^\load e^{-\budget} \Big( \budget \big(1-g(\load)\big)  - \int_{\budget}^{\load} \big( 1 - g(z) \big) dz \Big)^+ d\budget  + \int_\load^\infty e^{-\budget}  \budget \big(1-g(\load)\big) d\budget \ge \Gamma
        ~.
    \]

    We follow the same strategy as Huang and Zhang~\cite{HuangZ:STOC:2020} to get a numerical lower bound of the best achievable $\Gamma$ subject to the differential inequality.
    To guarantee above inequality, it suffices to find a pair of functions $g : [0, \infty) \to [0, 1]$ and $h : [0, \infty) \to [0, \infty)$ such that for any $\load \ge 0$:
    \[
        \int_0^\load e^{-\budget} g(\budget) d\budget + \int_{h(\load)}^\load e^{-\budget} \Big( \budget \big(1-g(\load)\big)  - \int_{\budget}^{\load} \big( 1 - g(z) \big) dz \Big) d\budget  + \int_\load^\infty e^{-\budget}  \budget \big(1-g(\load)\big) d\budget
        ~.
    \]

    Before we proceed further, let us first simplify the LHS.
    It is equal to:
    \begin{align*}
         &
        \int_0^\load e^{-\budget} g(\budget) d\budget - \int_{h(\load)}^\load e^{-\budget} \int_{\budget}^{\load} \big( 1 - g(z) \big) dz d\budget  + \int_{h(\load)}^\infty e^{-\budget}  \budget \big(1-g(\load)\big) d\budget \\
         & \qquad
        = \int_0^\load e^{-\budget} g(\budget) d\budget - \int_{h(\load)}^\load \big( e^{-h(\load)} - e^{-z} \big) \big( 1 - g(z) \big) dz + \big(1+h(\load)\big) e^{-h(\load)} \big(1-g(\load)\big)
        ~.
    \end{align*}

    Hence, we want to ensure that for any $\load \ge 0$:
    \[
        \int_0^\load e^{-\budget} g(\budget) d\budget - \int_{h(\load)}^\load \big( e^{-h(\load)} - e^{-z} \big) \big( 1 - g(z) \big) dz + \big(1+h(\load)\big) e^{-h(\load)} \big(1-g(\load)\big) \ge \Gamma
        ~.
    \]

    The rest of the proof relies on a computer-aided numerical optimization.
    Fixing any function $h$, we can find an optimal \emph{step-function} $g$ such that the above inequality holds for the largest possible $\Gamma$ by solving a linear program.
    By choosing a sufficiently small step, the resulting function $g$ and the corresponding $\Gamma$ are approximately optimal.
    Fixing any function $g$, on the other hand, finding the optimal function $h$ to maximize the LHS above is easy.
    Hence, we start with a trivial function $h(\ell) = 0$ for all $\ell$.
    Then, we alternate between (1) optimizing $g$ given $h$ and (2) choosing $h$ according to $g$ so that $h(\ell)$ is the smallest value for which:
    \[
        h(\ell) \big(1-g(\load)\big)  - \int_{h(\ell)}^{\load} \big( 1 - g(z) \big) dz \ge 0
        ~.
    \]

    Since the necessary discretization and the approach of alternating optimization were already detailed by Huang and Zhang~\cite{HuangZ:STOC:2020}, we merely provide a link to our code here.%
    \footnote{\href{https://github.com/dsgsjk/stochastic-balance}{https://github.com/dsgsjk/stochastic-balance}}
    The stated competitive ratio is achieved after only three rounds of alternating optimization.
\qed \end{proof}

\section{Missing Proofs from Section~\ref{sec:ranking}}
\label{app:ranking}

\subsection{Proof of Lemma~\ref{lem:ranking-non-stochastic-dual-feasibility-worst-order}}
\label{app:ranking-non-stochastic-dual-feasibility-worst-order}

We will prove the lemma by further relaxing and simplifying the right-hand-side of Eqn.~\eqref{eqn:ranking-non-stochastic-dual-feasibility-small}.
By the definition of $P$ and $\bar{P}$, the first two terms on the RHS of Eqn.~\eqref{eqn:ranking-non-stochastic-dual-feasibility-small} can be written as:
\begin{equation}
    \label{eqn:ranking-non-stochastic-dual-feasibility-bound-1st-term}
    \int_0^1 \big( 1 - e^{-P(\rank_u)-\bar{P}(\rank_u)} \big) g(\rank_u) \D\rank_u
\end{equation}
and:
\begin{equation}
    \label{eqn:ranking-non-stochastic-dual-feasibility-bound-2nd-term}
    - \int_0^1 \big(1 - g(\critical_v) \big) \D P(\critical_v)
    ~.
\end{equation}

We next examine the last term of Eqn.~\eqref{eqn:ranking-non-stochastic-dual-feasibility-original}, changing the order of integration and summation:
\begin{equation}
    \label{eqn:ranking-non-stochastic-dual-feasibility-bound-3rd-term}
    \int_0^1 \underbrace{\sum_{v \in N_u(\rank_u) \cap S} e^{-p|N_u(\rank_u,v)|} \big( 1 - e^{-p} \big) \big( g(\critical_v)-g(\rank_u) \big)}_{(\star)} \D \rank_u
    ~.
\end{equation}

Next decompose $N_u(\rank_u, v)$ into two parts to derive an upper bound:
\begin{align*}
    \big| N_u(\rank_u, v) \big|
     &
    = \big| N_u(\rank_u, v) \cap S \big| + \big| N_u(\rank_u, v) \setminus S \big| \\
     &
    \le \big| N_u(\rank_u, v) \cap S \big| + \big| N_u(\rank_u) \setminus S \big|
    ~,
\end{align*}
where it would hold with equality if $u$'s neighbors outside $S$ arrive before those in $S$.
Further recall that $\bar{P}(\rank_u) = p |N_u(\rank_u) \setminus S|$
We get:
\[
    (\star) ~\ge~ e^{- \bar{P}(\rank_u)} \big(1 - e^{-p}\big) \sum_{v \in N_u(\rank_u) \cap S} e^{- p |N_u(\rank_u, v) \cap S|} \, \big( g(\critical_v)-g(\rank_u) \big)
    ~.
\]

Observe that $|N_u(\rank_u, v) \cap S|$ simply counts the number of vertices in $N_u(\rank_u) \cap S$ that arrive before $v$.
Thus, when $v$ iterates over vertices in $N_u(\rank_u) \cap S$, $e^{- p |N_u(\rank_u, v) \cap S|}$ iterates over $e^{-pi}$ for $0 \le i < |N_u(\rank_u) \cap S|$.
By the rearrangement inequality, the RHS above is minimized when vertices in $N_u(\rank_u) \cap S$ arrive by the increasing order of their ranks.
Under this arrival order, and recalling that we have assumed distinct critical ranks, we have:
\[
    \big| N_u(\rank_u, v) \cap S \big| = \big| N_u(\rank_u) \cap S \big| - \big| N_u(\critical_v) \cap S \big|
    ~.
\]

Therefore, by $P(\critical) = p |N_u(\critical) \cap S|$ we have:
\begin{align*}
    (\star)
     &
    ~\ge~
    e^{- \bar{P}(\rank_u)} \big(1 - e^{-p}\big) \sum_{v \in N_u(\rank_u) \cap S} e^{- P(\rank_u) + P(\critical_v)} \big( g(\critical_v)-g(\rank_u) \big) \\
     &
    ~=~
    e^{- \bar{P}(\rank_u) - P(\rank_u)} \sum_{v \in N_u(\rank_u) \cap S} e^{P(\critical_v)} \big(1 - e^{-p}\big) \big( g(\critical_v)-g(\rank_u) \big)
    \\
     &
    ~=~
    - e^{- \bar{P}(\rank_u) - P(\rank_u)} \int_{\rank_u}^1 \big( g(\critical_v)-g(\rank_u) \big) \D e^{P(\critical_v)}
    ~.
\end{align*}

Putting this back to Eqn.~\eqref{eqn:ranking-non-stochastic-dual-feasibility-bound-3rd-term} shows that the third term on the RHS of Eqn.~\eqref{eqn:ranking-non-stochastic-dual-feasibility-small} is at least:
\begin{align*}
    - \int_0^1 e^{- \bar{P}(\rank_u) - P(\rank_u)} \int_{\rank_u}^1 \big( g(\critical_v)-g(\rank_u) \big) \D e^{P(\critical_v)}
    ~=~
     &
    - \int_0^1 e^{- \bar{P}(\rank_u) - P(\rank_u)} \int_{\rank_u}^1 g(\critical_v) \D e^{P(\critical_v)} d\rank_u
    \\
     & \quad
    -  \int_0^1 e^{- \bar{P}(\rank_u) - P(\rank_u)} \big( e^{P(\rank_u)} - 1 \big) g(\rank_u) \D \rank_u
    ~.
\end{align*}

Combining it with Equations~\eqref{eqn:ranking-non-stochastic-dual-feasibility-bound-1st-term} and \eqref{eqn:ranking-non-stochastic-dual-feasibility-bound-2nd-term} proves the lemma.
Specifically, merging the second term on the RHS above with Eqn.~\eqref{eqn:ranking-non-stochastic-dual-feasibility-bound-1st-term} yields the first term on the RHS of the lemma.

\subsection{Proof of Lemma~\ref{lem:ranking-non-stochastic-dual-feasibility-worst-inside}}
\label{app:ranking-non-stochastic-dual-feasibility-worst-inside}

We first rewrite the LHS of Eqn.~\eqref{eqn:ranking-non-stochastic-dual-feasibility-worst-size} in a discrete way.
Let $\critical_1\leq \critical_2\leq \cdots\leq\critical_n$ be the critical ranks of vertices in $S$.
Define:
\begin{align*}
    f(\critical_1,\critical_2,\dots,\critical_n)
     & = \int_0^{\critical_0} g(\rank_u) \D\rank_u + \sum_{i=1}^n p\bigg( 1-g(\critical_i) + \sum_{j=1}^i (\critical_j - \critical_{j-1}) e^{-p(i-j)} g(\critical_i)\bigg)                                      \\
     & = \int_0^{\critical_0} g(\rank_u) \D\rank_u + \underbrace{\vphantom{\sum_{j=1}^n}\sum_{i=1}^n p \big( 1-g(\critical_i) \big)}_{(a)} + \underbrace{\sum_{j=1}^n (\critical_j - \critical_{j-1}) \sum_{i=j}^n p e^{-p(i-j)} g(\critical_i)}_{(b)}
     ~.
\end{align*}


Let $\underline\critical\approx 0.513$ denote the smallest solution to $g(\rank)=1-\frac{1}{e}$. For any fixed $\critical_0$, we show that $f$ is minimized
either $\critical_1=\cdots=\critical_n=\underline\critical$ or $\critical_1=\dots=\critical_n=1$.

\begin{lemma}
    \label{left}
    For any $1 \le m < n$, any $\critical_1 < \critical_2$ that are between $\critical_0$ and $\underline{\critical}$, and any $\critical'_{m+1} \le \critical'_{m+2} \le \dots \le \critical'_n$ that are between $\critical_2$ and $1$, we have:
    \[
        f \big( \underbrace{\vphantom{\Big|} \critical_1,\dots,\critical_1}_{m},\critical'_{m+1},\dots,\critical'_n \big) ~\geq~ f\big( \underbrace{\vphantom{\Big|} \critical_2,\dots,\critical_2}_{m},\critical'_{m+1},\dots,\critical'_n \big)
        ~.
    \]
\end{lemma}

\begin{proof}
    Let $z = p \cdot m$ and recall that $n \cdot p = 1$.
    By the definition of $f$, the difference between the LHS and RHS equals:
    \begin{align*}
        &
        \underbrace{\vphantom{\sum_{i=m+1}^n} z \big(g(\critical_2)-g(\critical_1)\big)}_{\text{from $(a)$}} + \underbrace{(\critical_2-\critical_1) \sum_{i=m+1}^n p e^{-p(i-m-1)} g(\critical'_i)}_{\text{from $(b)$ when $j=m+1$}}
        \\
        &
        \underbrace{\vphantom{\sum_{i=m+1}^n} - \big(1-e^{-z}\big) \big(\critical_1-\critical_0\big) \big(g(\critical_2)-g(\critical_1)\big)}_{\text{from $(b)$ when $j = 1$}} - \underbrace{\vphantom{\sum_{i=m+1}^n} \big(\critical_2 - \critical_1\big) \Big( \big(1-e^{-z}\big) g(\critical_2) + \sum_{i=m+1}^n p e^{-p(i-1)} g(\critical'_i) \Big)}_{\text{from $(b)$ when $j = m+1$}}
        ~.
    \end{align*}

    Since the coefficients of $g(\critical'_i)$ is positive for all $m < i \le n$, and the coefficient of $\critical_0$ is positive, the difference is minimized when $\critical'_i = \critical_2$ and $\critical_0 = 0$.
    Hence, the difference is at least:
    \begin{align*}
        &
        z \big( g(\critical_2)-g(\critical_1) \big) + \big( \critical_2 - \critical_1 \big) \big(1 - e^{-(1-z)}\big) g(\critical_2) 
        \\
        & \qquad
        - \big(1-e^{-z}\big) \critical_1 \big( g(\critical_2)-g(\critical_1) \big) - \big(\critical_2 - \critical_1\big) \big(1 - e^{-1}\big) g(\critical_2)
        \\[1ex]
        & \quad
        = \big(z - (1-e^{-z}) \critical_1 \big) \big( g(\critical_2)-g(\critical_1) \big) - \big(e^{z-1}-e^{-1}\big)\big(\critical_2-\critical_1\big)g(\critical_2)
        ~.
    \end{align*}

    By the definition of $g$ and by that $\critical_1 < \critical_2 \le \underline{\critical}$, we have that for $\critical = \critical_1$ or $\critical_2$:
    \[
        g(\critical) = \frac{c}{e-(e-1)\critical}
        ~.
    \]

    Hence:
    \[
        g(\critical_2) - g(\critical_1) = (\critical_2 - \critical_1) \cdot \frac{c(e-1)}{(e-(e-1)\critical_1)(e-(e-1)\critical_2)}
        ~.
    \]

    Putting into the above bound we get that the difference is at least:
    \[
        \frac{c(\critical_2-\critical_1)}{e-(e-1)\critical_2} \Big( \underbrace{\frac{e-1}{e}\frac{z-(1-e^{-z})\critical_1}{1-(1-e^{-1})\critical_1}-(e^{z-1}-e^{-1})}_{(\star)} \Big)
        ~.
    \]

    It suffices to prove that $(\star) \ge 0$.
    Since $\frac{1-e^{-z}}{z}$ is decreasing in $z$ and in particular is at least $1-e^{-1}$, we get that $(\star)$ is decreasing in $\critical_1$.
    Further by that $\critical_1 \le \underline{\critical} < \frac{1}{e-1}$, it suffices to establish the $(\star)$ is non-negative when $\critical_1 = \frac{1}{e-1}$.
    In that case, it follows by:
    %
    %
    \[
        z-\frac{1-e^{-z}}{e-1}-(e^{z-1}-e^{-1})\geq 0
        ~,
    \]
    with equality achieved at $z = 0$ and $z = 1$
\qed \end{proof}

\begin{lemma}
    \label{right}
    For any $0 \le m < m' \leq n$, any $\critical_1 < \critical_2$ that are between $(\underline{\critical},1)$, 
    and any $\critical'_{1} \le \critical'_{2} \le \dots \le \critical'_m$ that are between $[0,\underline{\critical}]$, we have:
    \[
        f \big(\critical'_1,\dots,\critical'_m,\underbrace{\vphantom{\Big|} \critical_1,\dots,\critical_1}_{m'-m},\underbrace{\vphantom{\Big|} 1,\dots,1}_{n-m'}\big)\leq
        f \big(\critical'_1,\dots,\critical'_m,\underbrace{\vphantom{\Big|} \critical_2,\dots,\critical_2}_{m'-m},\underbrace{\vphantom{\Big|} 1,\dots,1}_{n-m'}\big)
        ~.
    \]
\end{lemma}

\begin{proof}
    Let $z = p \cdot (m'-m), w = p\cdot (n-m')$. By the definition of $f$, the difference between the LHS and RHS equals (note that $g(\mu_1)=g(\mu_2)=1-\frac{1}{e}$):
    \begin{align*}
        & \underbrace{\vphantom{\Big|} -(\critical_2-\critical_1)(1-e^{-z})g(\critical_2)-(\critical_2-\critical_1)e^{-z}(1-e^{-w})}_{\text{from $(b)$ when $j=m+1$}}
        + \underbrace{\vphantom{\Big|} (1-\critical_1)(1-e^{-w})-(1-\critical_2)(1-e^{-w})}_{\text{from $(b)$ when $j=m'+1$}} \\
        = & (1-e^{-w}-g(\critical_2))(\critical_2-\critical_1)(1-e^{-z})\\
        = & (1-e^{-w}-(1-\frac{1}{e}))(\critical_2-\critical_1)(1-e^{-z})\leq 0
    \end{align*}
\qed \end{proof}

\begin{lemma}
    \label{mid}
    $f(\underbrace{\vphantom{\Big|}\underline \critical,\dots,\underline \critical}_{m'},\underbrace{\vphantom{\Big|} 1,\dots,1}_{n-m'})$ 
    is minimized when $m'=0$ or $m'=n$ for any ${\underline \critical}$.
\end{lemma}

\begin{proof}
    Let $z = p\cdot (n-m')$ and define $F(z)=f(\underbrace{\vphantom{\Big|}\underline \critical,\dots,\underline \critical}_{m'},\underbrace{\vphantom{\Big|} 1,\dots,1}_{n-m'})$. By the definition of $f$,
    \begin{align*}
        F(z)=                                & (1-z)(1-g(\underline \critical))+(\underline \critical-\critical_0)((1-e^{z-1})g(\underline \critical)+(e^{z-1}-e^{-1})) \\
        +                                    & (1-\underline \critical)(1-e^{-z})+\int_0^{\critical_0} g(\rank_u) \D \rank_u                            \\
        \frac{\mathrm{d}F}{\mathrm{d}z}=     & -(1-g(\underline \critical))+(\underline \critical-\critical_0)e^{z-1}(1-g(\underline \critical))+e^{-z}(1-\underline \critical)  \\
        =                                    & (\underline \critical-\critical_0)e^{z-2}+e^{-z}(1-\underline \critical)-e^{-1}                                 \\
        \frac{\mathrm{d}^2F}{\mathrm{d}z^2}= & (\underline \critical-\critical_0)e^{z-2}-e^{-z}(1-\underline \critical)
    \end{align*}

    Since $\frac{\mathrm{d}^2F}{\mathrm{d}z^2}$ is increasing and $\frac{\mathrm{d}F}{\mathrm{d}z}(1)=-\frac{{\critical_0}}{e}\leq 0$, $F$ is minimized at either $z=0$ or $z=1$.
\qed \end{proof}

\begin{proof}[Proof of Lemma~\ref{lem:ranking-non-stochastic-dual-feasibility-worst-inside}]

    For vertices in $S$, assume $\critical_0\leq \critical_1\leq \dots \leq\critical_m\leq \underline\critical < \critical_{m+1}\leq \dots \leq \critical_{m'} < \critical_{m'+1}=\dots=\critical_n=1$.

    \begin{align*}
             & f(\critical_1,\dots,\critical_m,\critical_{m+1},\dots,\critical_{m'},\underbrace{\vphantom{\Big|}1,\dots,1}_{n-m'})                                                                                                        \\
        \geq & f(\underbrace{\vphantom{\Big|}\underline\critical,\dots,\underline\critical}_{m},\critical_{m+1},\dots,\critical_{m'},\underbrace{\vphantom{\Big|}1,\dots,1}_{n-m'}) &  & \text{(by an induction using Lemma~\ref{left})}  \\
        \geq & f(\underbrace{\vphantom{\Big|}\underline\critical,\dots,\underline\critical}_{m'},\underbrace{\vphantom{\Big|}1,\dots,1}_{n-m'})                                   &  & \text{(by an induction using Lemma~\ref{right})} \\
        \geq & \min(f(\underline\critical,\dots,\underline\critical),f(1,\dots,1))                                                                                   &  & \text{(by Lemma~\ref{mid})}
    \end{align*}

\qed \end{proof}

\end{document}